\documentclass{article} 
\usepackage{macros}
\usepackage{hyperref}
\usepackage{breakurl}

\title{Two Timescale Convergent Q-learning for Sleep--Scheduling in Wireless Sensor Networks}

\author[1]{Prashanth L A \thanks{prashanth.la@inria.fr}}
\author[2]{Abhranil Chatterjee \thanks{abhranilc@ee.iisc.ernet.in}}
\author[3]{Shalabh Bhatnagar \thanks{shalabh@csa.iisc.ernet.in}}
\affil[1]{\small INRIA Lille - Nord Europe, Team SequeL, FRANCE.}
\affil[2]{\small System Sciences and Automation,
Indian Institute of Science, Bangalore, INDIA.}
\affil[3]{\small Department of Computer Science and Automation,
Indian Institute of Science, Bangalore, INDIA}
\date{}

\begin{document}

\maketitle

\begin{abstract}
In this paper, we consider an intrusion detection application for Wireless Sensor Networks (WSNs). We study the problem of scheduling the sleep times of the individual sensors, where the objective is to maximize the network lifetime while keeping the tracking error to a minimum. We formulate this problem as a partially-observable Markov decision process (POMDP) with continuous state-action spaces, in a manner similar to \cite{fuemmeler2008smart}. However, unlike their formulation, we consider infinite horizon discounted and average cost objectives as performance criteria. 
For each criterion, we propose a convergent on-policy Q-learning algorithm that operates on two timescales, while employing function approximation. Feature-based representations and function approximation is necessary to handle the curse of dimensionality associated with the underlying POMDP. Our proposed algorithm incorporates a policy gradient update using a one-simulation simultaneous perturbation stochastic approximation (SPSA) estimate on the faster timescale, while the Q-value parameter (arising from a linear function approximation architecture for the Q-values) is updated in an on-policy temporal difference (TD) algorithm-like fashion on the slower timescale.
The feature selection scheme employed in each of our algorithms manages the energy and tracking components in a manner that assists the search for the optimal sleep-scheduling policy. 
For the sake of comparison, in both discounted and average settings, we also develop a function approximation analogue of the Q-learning algorithm. This algorithm, unlike the two-timescale variant, does not possess theoretical convergence guarantees. 
Finally, we also adapt our algorithms to include a stochastic iterative estimation scheme for the intruder's mobility model and this is useful in settings where the latter is not known. Our simulation results on a synthetic 2-dimensional network setting suggest that our algorithms result in better tracking accuracy at the cost of only a few additional sensors, in comparison to a recent prior work. 
\end{abstract}


\section{Introduction}

Considering the potential range of applications and the low deployment and maintenance cost, a lot of research attention has gone into the design of Wireless Sensor Networks (WSNs). In this paper, we investigate the use of WSNs for an intrusion detection application. In particular, we study the problem of scheduling the sleep times of the individual sensors, where the objective is to maximize the network lifetime while keeping the tracking error to a minimum. 

\tikzstyle{block} = [draw, fill=white, rectangle,
    minimum height=3em, minimum width=6em]
\tikzstyle{sum} = [draw, fill=white, circle, node distance=1cm]
\tikzstyle{input} = [coordinate]
\tikzstyle{output} = [coordinate]
\tikzstyle{pinstyle} = [pin edge={to-,thin,black}]
\tikzset{roads/.style={dotted,line width=0.05cm}}

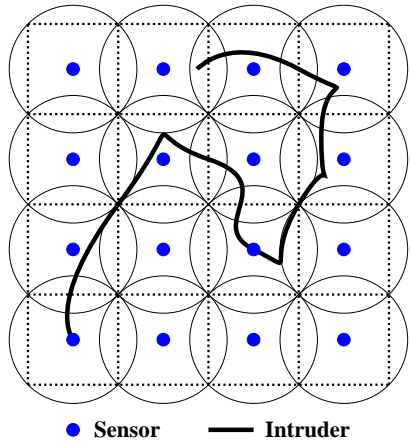
\begin{figure}[h]
\centering
\scalebox{0.6}{
  \tabl{c}{\begin{tikzpicture}
  \foreach \x in {0,2,4,6,8}
  {
  \draw[roads] (\x,0) -- (\x,8);
  \draw[roads] (0,\x) -- (8,\x);
  }
  
  \draw[line width=0.1cm] plot [smooth, tension=2] coordinates {  (1,1) (2,4) (4,5) (5,3) (6,4)  (6.5,5.5) (6,7) (3.75,7) };
  
  \foreach \x in {1,3,5,7}
  {
  \filldraw[blue] (1,\x) circle (0.14);
  \filldraw[blue] (3,\x) circle (0.14);
  \filldraw[blue] (5,\x) circle (0.14);
  \filldraw[blue] (7,\x) circle (0.14);
  
  \draw (1,\x) circle ({sqrt(2)});
  \draw (3,\x) circle ({sqrt(2)});
  \draw (5,\x) circle ({sqrt(2)});
  \draw (7,\x) circle ({sqrt(2)});
  }
  \filldraw[blue] (1,-1) circle (0.14);
  \node at (2.2,-1) {{\Large \bf Sensor}};
  \draw[line width=0.1cm] (4,-1) -- (5,-1);
  \node at (6.2,-1) {{\Large \bf Intruder}};
  
  \end{tikzpicture}\\[1ex]}}
\caption{Field of sensors and the movement of intruder considered}
\label{fig:2d}
\end{figure}

As illustrated in Fig. \ref{fig:2d}, we consider a centralized control setting for a sensor network involving $N$ sensors and assume for simplicity that the sensors fully cover the area of interest. Each sensor can be either awake (i.e., active) or asleep. The control center collects sensing information periodically and then decides on the sleeping policy for the sensors. The location of the intruder at any instant can be any one of the $N$ cells corresponding to the $N$ sensors.
The intruder movement is described by a Markov chain with a probability transition matrix $P$ of size $N\times N$. Each entry $P_{ij} \in [0,1]$ of the matrix $P$ specifies the probability of the intruder moving from location $i$ to $j$. The state of this Markov chain is the current location of the intruder to within the accuracy of a sensing region. The challenge is to balance the conflicting objectives of minimizing the number of sensors awake to reduce the energy cost, while at the same time having enough number of sensors awake to ensure a good tracking accuracy.

We formulate this problem as a partially-observable Markov decision process (POMDP), in a manner similar to \cite{fuemmeler2008smart}. However, unlike their formulation, we consider infinite horizon discounted and average cost objectives as performance criteria. The rationale behind the average cost objective is to understand the steady-state system behavior, whereas the discounted cost objective is more suitable for studying the transient behavior of the system. 

MDPs \cite{BertsekasDP01} (and POMDPs) are useful frameworks for modeling real-time control problems such as the sleep scheduling that we consider in this paper. However, in practice, the transition dynamics of the MDP is unavailable and reinforcement learning (RL) approaches provide an efficient alternative.  RL comprises of simulation-based sample-path techniques that converge to a good-enough policy in the long run. The reader is referred to \cite{BertsekasT96,sutton1998reinforcement} for a comprehensive (text book) introduction to RL.

We base our solution approach on reinforcement learning (RL) formalisms, with specific emphasis on developing Q-learning \cite{watkins1992q} type algorithms for learning the optimal sleeping policy for the sensors. At the same time, to be computationally efficient, we employ linear approximation architectures. Linear function approximation allows us to handle the curse of dimensionality associated with high-dimensional state spaces (as is the case with the sleep-wake scheduling problem considered in this paper). 
To the best of our knowledge, RL with function approximation
for sleep scheduling in WSNs has not been considered previously in the literature. 

However, for problems involving high-dimensional state spaces, the Q-learning algorithm with function approximation may diverge or may show large oscillations, \cite{baird}. This is primarily due to the inherent nonlinearity in the Q-learning update rule resulting from the explicit maximization/minimization in the update procedure. To alleviate this problem, we propose a two-timescale Q-learning algorithm, borrowing the principle of using a simultaneous perturbation method for policy gradient estimation from a closely related algorithm for the discounted setting proposed in \cite{bhatnagar2012twotimescale}.

Our algorithm, proposed for continuous state-action spaces and the long-run average cost criterion, operates on two timescales and works on the simultaneous perturbation principle \cite{spall1992multivariate,Bhatnagar13SR}. In particular, after parameterizing the policy in a continuous space, the algorithm updates the policy parameter along the negative gradient direction using a well-known simultaneous perturbation method called simultaneous perturbation stochastic approximation (SPSA) \cite[Chapter 5]{Bhatnagar13SR}. In particular, we employ a one-simulation SPSA estimate on the faster timescale for obtaining the policy gradients. On the other hand, along the slower timescale an on-policy TD-like update is performed for the Q-value parameters. This timescale separation together with the policy gradient estimation using SPSA gets rid of the off-policy problem present in vanilla Q-learning with function approximation. 
The resulting algorithm turns out to be a stochastic approximation scheme on the faster timescale, but a stochastic recursive inclusion \cite[Chapter 5]{borkar2008stochastic} scheme on the slower timescale. We provide a sketch of the convergence of this algorithm, with the detailed proof being available in an appendix to this paper. To the best of our knowledge, a convergent Q-learning type algorithm with function approximation to optimize a long-run average cost criterion in a POMDP with continuous state-action spaces (as is the case with the sleep-scheduling POMDP considered), has not been proposed earlier in the literature.


We summarize our contributions as follows\footnote{A short version of this paper containing only the average cost setting and algorithms and with no proofs is available in \cite{comsnets}. The current paper includes in addition: 
\begin{inparaenum}[\bfseries(i)]
\item algorithms for the discounted cost setting; 
\item a detailed proof of convergence of the average cost algorithm using theory of stochastic recursive inclusions; and
\item detailed numerical experiments.
\end{inparaenum}}
:\\
\begin{inparaenum}[\bfseries(i)]
\item In the average cost setting, we propose a novel two-timescale algorithm that performs on-policy Q-learning while employing function approximation. This algorithm is efficient owing to linear function approximators and possesses theoretical convergence guarantees. 
For the sake of comparison, we also develop a function approximation analogue of the Q-learning algorithm. This algorithm, unlike the two-timescale variant, does not possess theoretical convergence guarantees. 
The feature selection scheme employed in each of our algorithms manages the energy and tracking components in a manner that assists the search for the optimal sleep-scheduling policy. \\
\item In the discounted setting, we adapt the recently proposed two-timescale (convergent) variant of the Q-learning algorithm \cite{bhatnagar2012twotimescale}, with function approximation. Further, for the sake of comparison, we also develop a sleep-scheduling algorithm based on Q-learning with linear function approximation. These algorithms can be seen to be the discounted-cost counterparts of the algorithms described above for the average cost setting. \\
\item We also adapt our algorithms to a setting where the mobility model of the intruder is not available. We develop a stochastic iterative scheme that estimates the mobility model and combine this estimation procedure with the average cost algorithms mentioned above using multi-timescale stochastic approximation.\\
\item We validate our algorithms on a two-dimensional network setting, while also comparing their performance with the \qmdp and FCR algorithms from \cite{fuemmeler2008smart}. Our algorithms are seen to be easily implementable, converge rapidly with a short (initial) transient period and provide more consistent results than the \qmdp and FCR algorithms.  Further, we observe that the procedure for estimating the mobility model of the intruder converges empirically to the true model. 
\end{inparaenum}

The rest of the paper is organized as follows: In Section \ref{sec:literature-review}, we review relevant literature in the area of sleep-wake scheduling as well as reinforcement learning. In Section \ref{sec:problem-formulation}, we formulate the sleep-wake scheduling problem as a POMDP.  and describe the long-run performance objectives (both average and discounted) for our algorithms. In Section \ref{sec:average-setting} we describe the long-run average cost objective and in Section \ref{sec:average-algorithms}, we present two novel RL-based sleep-wake scheduling algorithms for this setting. In Section \ref{sec:unknownP}, we present the mobility model estimation scheme. In Section \ref{sec:discounted-objective} we present the discounted cost objective and in Section \ref{sec:discounted-algorithms}, we extend the average cost algorithms to this setting.  In Section \ref{sec:simulations}, we describe the experimental setup and present the results in both average and discounted cost settings. Finally, in Section \ref{sec:conclusions}, we provide the concluding remarks and outline a few future research directions.

\section{Related Work}
\label{sec:literature-review}
Sleep scheduling is broadly related to the problem of resource allocation in wireless networks. A comprehensive survey of solution approaches, including RL-like schemes, is available in \cite{cui2012survey}. Further, considering this problem from a strategic, i.e., game-theoretic, perspective, the authors in \cite{fu2009learning} propose an auction based best-response algorithm.
A two-timescale stochastic approximation algorithm for downlink scheduling in a cellular wireless system is proposed in \cite{cui2012delay}.

 In \cite{premkumar2008optimal}, the authors formulate an MDP model for intrusion detection and present algorithms to control the number of sensors in the wake state. 
In \cite{liu2006rl,niu2010self,jianlin2009rl}, the authors propose RL based medium access control (MAC) protocols for WSNs. The algorithms proposed there attempt to maximize the throughput while being energy efficient. In \cite{liu2006rl,niu2010self}, the authors propose  Q-learning based algorithms, whereas, in \cite{jianlin2009rl}, the authors propose an algorithm based on SARSA.
In \cite{gui2004power}, the authors present two sleep scheduling algorithms for single object tracking.   In \cite{jiang2008energy}, a sleep scheduling algorithm based on the target's moving direction has been proposed. In \cite{jin2006dynamic}, the authors present a heuristic algorithm that uses dynamic clustering of sensors to  balance energy cost and tracking error. 
In \cite{beccuti2009multiple}, the problem of finding an efficient sleep-wake policy for the sensors while maintaining good tracking accuracy by solving an MDP has been studied.
In \cite{khan2012resource}, the authors propose a Q-learning based algorithm for sleep scheduling. 
In \cite{fuemmeler2008smart,fuemmeler2011sleep}, the authors propose a POMDP model for sleep-scheduling in an object tracking application and propose several algorithms based on traditional dynamic programming approaches to solve this problem. 

In comparison to previous works, we would like to point out the following:\\
\begin{inparaenum}[\bfseries(i)]
\item Some of the previously proposed algorithms, for instance \cite{premkumar2008optimal}, require the knowledge of a system model and this may not be available in practice. On the other hand, our algorithms use simulation-based values and optimize along the sample path, without necessitating a system model.
\item Some algorithms, for instance \cite{gui2004power}, work under the waking channel assumption, i.e., a setting where the central controller can communicate with a sensor that is in the sleep state. Our algorithm do not operate under such an assumption.
\item In comparison to the RL based approaches \cite{liu2006rl,niu2010self,jianlin2009rl} for transmission scheduling at the MAC layer, we would like to point out that the algorithms proposed there  
\begin{inparaenum}[\bfseries(a)]
\item employ full state representations;
\item consider discrete state-action spaces (except \cite{niu2010self} which adapts Q-learning for continuous actions, albeit with a discrete state space);
\item consider an MDP with perfect information, i.e., a setting where the states are fully observable;
\item consider only a discounted setting, which is not amenable for studying steady state system behaviour;
\item are primarily concerned with managing transmission in an energy-efficient manner and not with tracking an intruder with high-accuracy.
\end{inparaenum}  
In other words, the algorithms of \cite{liu2006rl,niu2010self,jianlin2009rl} are not applicable in our setting as we consider a partially observable MDP with continuous state-action spaces, and with the aim of minimizing a certain long-term average cost criterion that involves the conflicting objectives of reducing energy consumption and maintaining a high tracking accuracy.
\item Many RL based approaches proposed earlier for sleep scheduling (see \cite{liu2006rl,niu2010self,jianlin2009rl,rucco2013bird}) employ full state representations and hence, they are not scalable to larger networks owing to the curse of dimensionality. We employ efficient linear approximators to alleviate this.
\item While the individual agents in \cite{fu2009learning} employ a RL-based bidding scheme, their algorithm is shown to work well only empirically and no theoretical guarantees are provided. This is also the case with many of the earlier works on sleep scheduling/power management in WSNs using RL and this is unlike our two-timescale on-policy Q-learning based scheme that possesses theoretical guarantees.
\item In \cite{cui2012delay}, the authors derive an equivalent Bellman Equation (BE) after reducing the state space and establish convergence of their algorithm to the fixed point of the equivalent BE. However, there is no straightforward reduction of state space in our sleep scheduling problem and we employ efficient linear function approximators to alleviate the curse of 
dimensionality associated with large state spaces.    
\item In comparison to \cite{fuemmeler2008smart}, which is the closest related work, we would like to remark that the algorithms proposed there, for instance, \qmdp, attempt to solve a balance equation for the total cost in an
approximate fashion at each time instant and no information about
the solution thus obtained is carried forward to the future instants. Moreover, unlike \cite{fuemmeler2008smart}, we consider long-run performance objectives that enable us to study both the transient as well as steady state system behavior.
\end{inparaenum}

In general, we would like to remark that unlike previous works on sleep-scheduling, we propose RL-based algorithms that
observe the samples of a cost function from simulation and through incremental
updates find a `good enough' policy that minimizes the long-run (average or discounted) sum of this cost. The term `good enough' here refers to the solution of a balance equation for the long term costs, where function
approximation is employed to handle the curse of dimensionality. Our algorithms are simple, efficient and in the case of the two-timescale on-policy Q-learning based schemes, also provably convergent.


\section{POMDP Formulation}
\label{sec:problem-formulation}

The state $s_k$ at instant $k$ for our problem is $s_k = (l_{k},r_{k})$, where $r_k = (r_k(1), \ldots, r_k(N))$,
is the vector of residual (or remaining) sleep times, with $r_k(i)$ denoting the residual sleep time of sensor $i$ at time instant $k$. Further, $l_k$ refers to the location of the object at instant $k$ and can take values $1,\ldots,N$.
The residual sleep time vector $r_k$ evolves as follows: $\forall i=1,\ldots,N$,
\begin{equation}
\label{revolve}
 r_{k+1}(i)=(r_{k}(i)-1)\I_{\{r_{k}(i) > 0\}} + a_{k}(i)\I_{\{r_{k}(i) = 0\}}.
\end{equation}
In the above $\I_{\{C\}}$ denotes the indicator function, having the value $1$ when the condition $C$ is true and $0$ otherwise.
The first term in \eqref{revolve} indicates that the residual sleep time is decremented by $1$ if sensor $i$ is in sleep state, while the second term expresses that if sensor $i$ is in wake state, it is assigned a sleep time of $a_k(i)$. Here $a_k = (a_k(1), \ldots, a_k(N))$ denotes the chosen sleep configuration of the $N$ sensors at instant $k$. 

The single-stage cost function has two components - an energy cost for sensors in the wake state and a tracking cost. We use an energy cost $c \in (0,1)$  for each sensor that is awake and a tracking cost of $1$  if the intruder location is unknown.
Let $S_k$ denote the set of indices of sensors that are in sleep state.
Then the single-stage cost $g(s_k,a_k)$ at instant $k$ has the form,
\begin{equation}
\label{eq:single-stage-cost}
 g(s_{k},a_{k})=\sum_{\{i: r_{k}(i)=0\}}c + \mathcal{I}_{\{r_{k}(l_{k}) > 0\}}.
\end{equation}
Since the number of sensors is finite, the single-stage cost is uniformly bounded.
 The algorithms that we design subsequently find the optimal strategy for minimizing the single-stage cost \eqref{eq:single-stage-cost} in the long-run average cost sense. Note that, unlike the formulation in \cite{fuemmeler2008smart}, we do not consider a special termination state which indicates that the intruder has left the system\footnote{Since we study long-run average sum of \eqref{eq:single-stage-cost} (see \eqref{eq:averagecost} below), we can consider the problem of tracking an intruder in an infinite horizon, whereas a termination state in \cite{fuemmeler2008smart} was made necessary as they considered a total cost objective.}.   

The states, actions and single-stage cost function together constitute an MDP.
However, since it is not possible to track the intruder at each time instant (i.e., $l_k$ is not known for all $k$) as the sensors at the location from where the intruder passes at a given time instant may be in the sleep state, the problem falls under the realm of MDPs with imperfect state information, or alternatively partially observed MDP (POMDP). Following the notation from \cite{fuemmeler2008smart}, the observation $z_k$ available to the control center is given by $z_k = (s_k, o_k)$, where $s_k$ is as before and $o_k=l_k$ if the intruder location is known, or a special value $\zeta$ otherwise. Thus, the total information available to the control center at instant $k$ is given by $I_k=(z_0,\ldots,z_k,a_0,\ldots,a_{k-1})$,
where $I_0$ denotes the initial state of the system. The action $a_k$ specifies the chosen sleep configuration of the $n$ sensors and is a function of $I_k$. As pointed out in \cite{fuemmeler2008smart}, in the above POMDP setting, a sufficient statistic is $\hat s_k = (p_k, r_k)$, where $p_{k} = P(\left. l_k \right| I_k)$ and $r_k$ is the remaining sleep time mentioned above. Note that $p_{k}=(p_{k}(1),...,p_{k}(N))$ is the distribution at time step $k$ of the object being in one
of the locations $1,2,...,N$ and evolves according to
\begin{equation}
\label{pevolve}
 p_{k+1}= e_{l_{k+1}}\mathcal{I}_{\{r_{k+1}(l_{k+1})=0\}} + p_{k} P \mathcal{I}_{\{r_{k+1}(l_{k+1}) > 0\}},
\end{equation}
where $e_i$ denotes an $N$-dimensional unit vector with $1$ in the $i$th position and $0$ elsewhere.
The idea behind the evolution of $p_k$ is as follows:\\
\begin{inparaenum}[\bfseries(i)] 
\item the first term refers to the case when the location of the intruder is known, i.e., the sensor at $l_{k+1}$ is in the wake state; \\
\item the second term refers to the case when intruder's location is not known and hence, the intruder transitions to the next distribution $p_{k+1}$ from the current $p_k$ via the transition probability matrix $P$.\\ 
\end{inparaenum} 
Note that the evolution of $p_k$ in our setting differs from \cite{fuemmeler2008smart}, as we do not have the termination state.
With an abuse of terminology, henceforth we shall refer to the sufficient statistic $\hat s_k$ as the state vector in the algorithms we propose next. Further, we would like to emphasize here that our algorithms do not require full observation of the state vector. Instead, by an intelligent choice of features that rely only on $p_k$, the algorithms obtain a sleeping policy that works well.

\section{Average Cost Setting}
\label{sec:average-setting}
The long-run average cost $J(\pi)$ for a given policy $\pi$ is defined as follows:
\begin{equation}
\label{eq:averagecost}
J(\pi) = \lim_{N \rightarrow \infty}\frac{1}{N} \sum_{n=0}^{N-1} g(s_n,a_n),
\end{equation}
starting from any given state $i$ (i.e., with $s_0 = i$). In the above, the policy $\pi=\{\pi_0$, $\pi_1$, $\pi_2, \ldots \}$ with 
$\pi_n$ governing the choice of action $a_n$ at each instant $n$.

The aim here is to find a policy $\pi^*= \mathop{{\rm argmin}}_{\pi \in \Pi} J(\pi)$, where $\Pi$ is the set of all admissible policies. A policy $\pi$ is admissible if it suggests a feasible action at each time instant $n$.

Let $h(x)$ be the differential cost function corresponding to state $x$, under policy $\pi$. Then, 
\begin{equation}
h(x) = \sum\limits_{n=1}^{\infty} E\left[ g(s_n,a_n) - \left. J(\pi) \right| s_0=x, \pi \right],
\label{vf-a}
\end{equation}
is the expected sum of the differences between the single-stage cost and the average cost under policy $\pi$ when $x \in \S$ is the initial state. 
Let $J^* = \min_{\pi \in \Pi} J(\pi) \stackrel{\triangle}{=} J(\pi^*)$ denote the optimal average cost and let $h^*$ denote the optimal differential cost function corresponding to the policy $\pi^*$. Then, $(J^*,h^*(x)),x\in \S$ satisfy the following Bellman equation (see \cite{puterman}):
\begin{equation}
J^* + h^*(x) = \min_a ( g(x,a) + \int p(x,a,dy) h^*(y) ), \forall x \in \S,
\label{bf-a}
\end{equation}
where $p(x,a,dy)$ denotes the transition probability kernel of the underlying MDP. 
Now, define the optimal Q-factors $Q^*(x,a) , x \in \S, a \in \A(x)$ as
\begin{equation}
 Q^*(x,a) =  g(x,a) + \int p(x,a,dy) h^*(y).
\label{qf}
\end{equation}
From \eqref{bf-a} and \eqref{qf}, we have
\begin{equation}
J^* + h^*(x) = \min_a Q^*(x,a), \forall x \in \S.
\label{qf_1}
\end{equation}
Now from \eqref{qf} and \eqref{qf_1}, we have
\begin{align}
 Q^*(x,a) = g(x,a) + \int p(x,a,dy)(\min_b Q^*(y,b) - J^*) & \mbox{ or}
\nonumber\\
J^* + Q^*(x,a) =   g(x,a) + \int p(x,a,dy)\min_{b} Q^*(y,b),&
\label{be-a}
\end{align}
for all $x \in \S, a \in \A(x)$.
An advantage with \eqref{be-a} is that it is amenable to stochastic approximation because the minimization is now (unlike \eqref{bf-a}) inside the conditional expectation.
However, in order to solve \eqref{be-a}, one requires knowledge of the transition kernel $p(x,a, dy)$ that
constitutes the system model. Moreover, one requires the state and action spaces to be manageable in size.
The algorithms presented subsequently work under lack of knowledge about the system model  and further, are able to effectively handle large state and action spaces by incorporating feature based representations and function approximation.

For the two-timescale on-policy Q-learning scheme (TQSA-A), we consider a parameterized set of policies that satisfy the following assumption:
\begin{assumption}
\label{a1}
For any state-action pair $(x,a)$, $\pi_w(x,a)$ is continuously differentiable in the parameter $w$.
\end{assumption}
The above is a standard assumption in policy gradient RL algorithms (cf. \cite{bhatnagar2009natural}).
A commonly used class of distributions that satisfy this assumption for the policy $\pi$ is the parameterized Boltzmann family, where the distributions have the form
\begin{equation}
\pi_{w}(x,a) = \frac{e^{w^{\top} \sigma_{x,a}}}{\sum_{a' \in {A(x)}} e^{w^{\top} \sigma_{x,a'}}},
\hspace{6pt} \forall x \in \S\;,\;\forall a \in \A(x).
\label{eq:pi}
\end{equation}
In the above, the parameter $w=(w_1,\ldots,w_N)^T$ is assumed to take values in a compact and convex set $C \subset \mathbb{R}^{N}$. 

Before we proceed further, it is important to note there that in our setting, we have a continuous state-action space. Hence, to implement our Q-learning algorithms, we discretize the space to a finite grid (as is commonly done in practice). In what follows, we shall consider $(x,a), (y,b)$ to take values on the aforementioned finite grid of points and $p(x,a,y)$ to denote the transition probabilities of the resulting Markov chain.

\section{Average Cost Algorithms}
\label{sec:average-algorithms}

For the ease of exposition, we first describe the Q-learning algorithm that uses full-state representations. Next, we discuss the difficulty in using this algorithm on a high-dimensional state space (as is the case with the sleep-wake control MDP) and subsequently present our average cost algorithms that employ feature based representations and function approximation to handle the curse of dimensionality.

\subsection{Q-learning with full state representation}
This algorithm, proposed in \cite{abounadi2002learning}, is based on the relative Q-value iteration (RQVI) procedure.
Let $s_{n+1}$ denote the state of the system at instant $(n+1)$ when the state at instant
  $n$ is $x$ and action chosen is $a$. Let $Q_{n}(x,a)$ denote the Q-value estimate at instant $n$
  associated with the tuple $(x,a)$.
The RQVI scheme (assuming a finite number of state-action tuples)
\begin{align}
Q_{n+1}(x,a)  = &  g(x,a) + \sum_{y} p(x,a,y) \underbrace{\min_{b\in \A(y)} Q_n(y,b)}_{\bf (I)}  - \underbrace{\min_{r \in \A(s)} Q_n(s,r)}_{\bf (II)},
\label{rvi}
\end{align}
where $s \in S$ is a prescribed (arbitrarily chosen) state\footnote{A simple rule to choose a state $s$ such that there is a positive probability of the underlying MDP visiting $s$. Such a criterion ensures that the term (II) of \eqref{rvi} converges to the optimal average cost $J^*$.}. 
Note, unlike the value iteration scheme for discounted MDPs, the recursion \eqref{rvi} includes an additional term (see (II) in \eqref{rvi}). This term arises due to the nature of the Bellman equation for average cost MDPs (see \eqref{be-a}) that also contains the optimal average cost $J^*$.
Here the state $s$ can be arbitrarily chosen because one is interested in estimating not just the average cost, but also the differential cost function. This results in solving a system of $(n+1)$ unknowns using $n$ equations. In order to make this system feasible, one fixes the differential cost for one of the (arbitrarily chosen) state to be a fixed value and then solves for the remaining $n$ values using the system of $n$ equations.
 It has been shown in \cite{abounadi2002learning} that term (II) in \eqref{rvi} converges to $J^*$ and term (I) in \eqref{rvi} converges to the optimal differential cost function $h^*(\cdot)$.
  
The Q-learning algorithm for the average cost setting estimates the `Q-factors' $Q(x,a)$ of all feasible state-action tuples $(x,a)$,
  i.e., those with $x \in S$ and $a \in \A(x)$ using the stochastic approximation version of \eqref{rvi}. The update rule for this algorithm is given by
  \begin{align}
Q_{n+1}(x,a)  = & Q_{n}(x,a) + a(n) ( g(x,a) + \min_{b\in A(y)} Q_n(y,b)  - \min_{r \in A(s)} Q_n(s,r) ),
\label{eqn:qtlcfsac}
\end{align}
for all $x \in S$ and $a \in\ A(x)$.
In the above, $y$ is the simulated next state after $x$ when action $a$ is chosen in state $x$ and $a(n), n \ge 0$ are the step-sizes that satisfy the standard stochastic approximation conditions, i.e., $\sum_n a(n) = \infty$ and $\sum_n a(n)^2 < \infty$.
The last term $\min_{r \in A(s)} Q_n(s,r)$ in (\ref{eqn:qtlcfsac}) asymptotically converges to the optimal average cost per stage.
Further, the iterates in \eqref{eqn:qtlcfsac} converge to the optimal Q-values $Q^{*}(i,a)$ that satisfy the corresponding
Bellman equation (\ref{be-a}) and $\min_{a \in A(i)} Q_n(i,a)$ gives the optimal differential cost $h^*(i)$. The optimal action in state $i$ corresponds to
$ \mathop{{\rm argmin}}_{a\in A(i)} Q^{*}(i,a)$. 

\subsection{Need for function approximation}
While Q-learning does not require knowledge of the system model, it does suffer from the computational problems associated with
large state and action spaces as it stores the $Q(s,a)$ values in a look-up table and requires updates of all $Q(s,a)$ values at
each step for convergence. In our setting, this algorithm  becomes intractable as the state-action space becomes very large. Even when we
quantize probabilities as multiples of $0.01$, and with $7$ sensors, the cardinality of the state-action space
$|S \times A(S)|$ is approximately $100^{8} \times 4^{7} \times 4^{7}$ if we use an upper bound of $3$ for the sleep time alloted to any sensor.  The situation gets aggravated
when we consider larger sensing regions (with corresponding higher number of sensors). To deal with this problem of the curse of
dimensionality, we develop a feature based Q-learning algorithm as in \cite{la2011reinforcement}.
While the full state Q-learning algorithm in (\ref{eqn:qtlcfsac}) cannot be used on even moderately sized sensing
regions, its function approximation based variant can be used over larger network settings.

\subsection{Algorithm Structure}
Both our algorithms parameterize the Q-function using a linear approximation architecture as follows:
\begin{equation}
 Q(s,a) \approx \theta^T \sigma_{s,a}, \quad \forall s \in S, a \in A(s).
 \label{eqn:Q_FA}
\end{equation}
In the above, $\sigma_{s,a}$ is a given  $d$-dimensional feature vector associated with the state-action tuple $(s,a)$, where  $d << |S \times A(S)|$ and $\theta$ is a tunable $d$-dimensional parameter. The Q-value parameter $\theta = (\theta_1, \ldots, \theta_d)^T$ is assumed to take values in a compact and convex set $D \subset \mathbb{R}^d$.

Our algorithms are online, incremental and obtain the sleeping policy by sampling from a trajectory of the system. After observing a simulated sample of the single-stage cost, the parameter $\theta$ is updated in the negative descent direction in both our algorithms as follows:
\begin{equation}
\label{eqn:theta_estimate}
\theta_{n+1} = \Gamma(\theta_n - a(n) \sigma_{s_n,a_n} m_n),
\end{equation}
where $m_n$ is an algorithm-specific magnitude term and $\Gamma$ is a projection operator that keeps the parameter $\theta$ bounded (a crucial requirement towards ensuring convergence of the scheme). Further, $a(n)$ are the step-sizes that satisfy standard stochastic approximation conditions.
Note that $\nabla_\theta Q(s,a) = \sigma_{s,a}$ and hence \eqref{eqn:theta_estimate} updates the parameter $\theta$ in the negative descent direction. 
The overall structure of our algorithms is given in Algorithm \ref{algorithm:structure}.

\begin{algorithm}
\caption{Structure of our algorithms}
\label{algorithm:structure}
\begin{algorithmic}[1]
\STATE {\bf Initialization:} policy parameter $\theta=\theta_0$; initial state $s_0$
\FOR{$n = 0,1,2,\ldots$}
  \STATE {Take action $a_{n}$ based on a (algorithm-specific) policy depending on $\theta_n$.}
  \STATE {Observe the single-stage cost $g(s_{n},a_{n})$ and the next state $s_{n + 1}$.}
  \STATE {Update $\theta_{n + 1}$ in a algorithm-specific manner.}
\ENDFOR
\STATE {\bf return} Q-value parameter $\theta$, policy parameter $w$.
\end{algorithmic}
\end{algorithm}

\subsection{Feature selection}
\label{sec:features}
\begin{figure}
\centering
\begin{tikzpicture}
    \begin{axis}[
        width=10cm, height=7cm,     
        grid = major,
        grid style={dashed, gray!30},
        xmin=0,     
        xmax=1,    
        ymin=0,     
	yticklabels={,,},
	xticklabels={,,},
        axis background/.style={fill=white},
        ylabel={\large Cost $\in (0,1)$},
        xlabel={\large Sleep time (seconds)},
        tick align=outside]
      \addplot[domain=0:1, blue, thick] 
         {e^(-3*x)}; 
         \addlegendentry{Energy cost}
      \addplot[domain=0:1, red, thick] 
         {1-e^(-3*x)}; 
         \addlegendentry{Tracking error}
      \addplot[domain=0:1, black, dotted,very thick]
      {0.55};
      \addplot[domain=0:1, black, dotted,very thick]
      {0.45};
      \addplot [black,  nodes near coords={\Large{$\bm\xi$}},every node near coord/.style={anchor=180}] coordinates {( 0.75, 0.5)};
      \addplot [black, nodes near coords={\Large{$\bm\downarrow$}},every node near coord/.style={anchor=180}] coordinates {( 0.75, 0.6)};
      \addplot [black, nodes near coords={\Large{$\bm\uparrow$}},every node near coord/.style={anchor=180}] coordinates {( 0.75, 0.4)};
    \end{axis} 
\end{tikzpicture}
\caption{Idea behind the feature selection scheme}
\label{fig:features}
\end{figure}
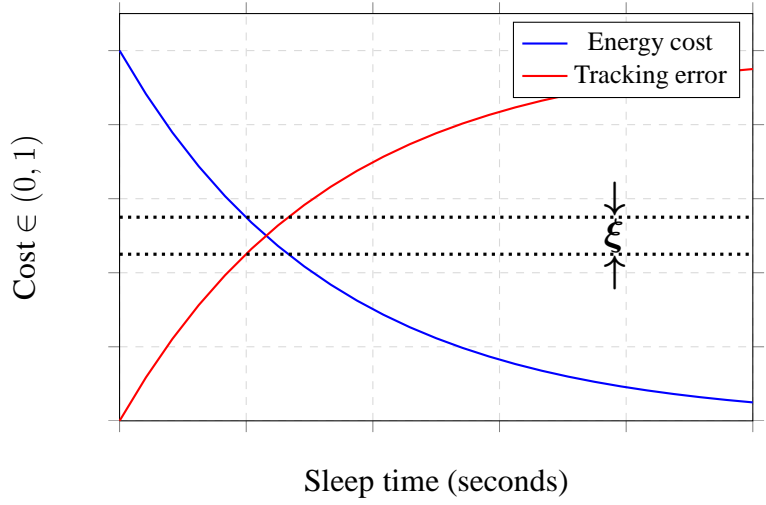

The idea behind the feature selection scheme is to select an energy-efficient sleep configuration, i.e., a configuration that  keeps as many sensors in the wake state as possible to track the intruder while at the same time has minimal energy cost.
This is done by first pruning the actions so as to select only those actions that ensure that the energy cost is $\xi$-close to the tracking error and then, among the $\xi$-optimal actions, selecting an action that minimizes the approximate Q-value. 

Formally, the choice of features is given by
\begin{equation}
\sigma_{s_n,a_n}  =  (\sigma_{s_{n},a_{n}}(1),...,\sigma_{s_{n},a_{n}}(N))^T,
\end{equation}
where $\sigma_{s_{n},a_{n}}(i), i \le N$ is the feature value corresponding to sensor $i$. These values are defined as follows:
\begin{align}
\label{eq:rho}
\delta^{a_n(i)}_n=\underbrace{\dfrac{1}{(a_{n}(i)+1)}}_{{\bf energy cost}}-\underbrace{\dfrac{\sum_{j=1}^{a_{n}(i)}[pP^{j}]_{i}}{\sum_{j=1}^{\infty}[pP^{j}]_{i}}}_{{\bf tracking error}},\\[1ex]
\label{eq:features}\sigma_{s_n,a_n}(i)=
\begin{cases}  \delta^{a_n(i)}_n &\text{if }0 \leq |\delta^{a_n(i)}_n| \leq \xi,
\\
\top &\text{ otherwise.}
\end{cases}
\end{align}
In the above, $\top$ is a fixed large constant used to prune out the actions that are not $\xi$-close.
The above choice of features involve pruning of actions, which is explained as follows:
Consider an action $a_{n}(i)$ for the sensor $i$ at time instant $n$.
The sum of probabilities that the intruder will be at location $i$, over time instants $1,\ldots, a_n(i)$ is a measure of the tracking error. On the other hand, the energy saved by having sensor $i$ sleep for $a_n(i)$ time units is proportional to $\frac{c}{a_{n}(i)+1}$. As illustrated in Fig. \ref{fig:features}, the tracking error increases with the sleep time (dictated by the choice of $a_{n}(i)$), while the energy cost decreases. Thus, $\delta^{a_n(i)}_n$ measures the distance between the energy cost and tracking errors. Next, as illustrated with the two-dashed lines in Fig. \ref{fig:features}, we now consider all those actions $a_n(i)$ such that the above two components are within $\xi$ distance of each other (i.e., $|\delta^{a_n(i)}_n| \leq \xi$) and set the feature value $\sigma_{s_n,a_n}$ to the above difference. On the other hand, for those actions that are outside the $\xi$-boundary, we set $\sigma_{s_n,a_n}$ to a large constant, which ensures they are not selected. 

In the following section, we present the QSA-A algorithm for sleep-wake scheduling and subsequently present the second algorithm (TQSA-A).  The latter algorithm (TQSA-A) is a convergent algorithm, unlike QSA-A.

\subsection{Q-learning based Sleep--wake Algorithm (QSA-A)}
\label{sec:qsa-a}
This is the function approximation analogue of the Q-
learning with average cost algorithm \cite{abounadi2002learning}.
Let $s_n, s_{n+1}$ denote the state at instants $n, n+1$, respectively, measured online. Let
$\theta_n$ be the estimate of the parameter $\theta$ at instant $n$. Let $s$ be any fixed state in $S$.
The algorithm QSA-A uses the following update rule:
\begin{align}
\theta_{n+1}  = & \theta_n + a(n) \sigma_{s_n,a_n}\bigg(g(s_n,a_n) + \underbrace{\min_{v\in A(s_{n+1})} \theta_n^T \sigma_{s_{n+1},v}}_{\bf (I)} 
                  - \underbrace{\min_{r\in A(s)} \theta_n^T \sigma_{s,r}}_{\bf (II)} - \theta_n^T \sigma_{s_n,a_n} \bigg),
\label{eq:qsa-a-update-rule}
\end{align}
where $\theta_0$ is set arbitrarily. In (\ref{eq:qsa-a-update-rule}), the action $a_n$ is chosen in state $s_n$ according to
an $\epsilon$-greedy policy, i.e., with a probability of ($1-\epsilon$), a greedy action given by\\
$a_{n} = \mathop{{\rm argmin}}_{v\in A(s_n)} \theta_n^T \sigma_{s_{n},v}$ is chosen and with probability $\epsilon$, an action in $A(s_n)$ is randomly chosen.
Using $\epsilon$-greedy policy for the regular Q-learning algorithm has been well recognized and recommended in the literature (cf. \cite{sutton1998reinforcement,BertsekasT96}). 

\subsection{Two-timescale Q-learning based sleep--wake algorithm (TQSA-A)}
\label{sec:tqsa-a}
Although Q-learning with function approximation has been shown to work well in several applications in practice, establishing a proof of convergence of this algorithm is theoretically difficult. A simple counterexample that illustrates the chattering phenomenon when Q-learning is combined with function approximation is provided in \cite[Section 6.4]{BertsekasT96}. Moreover, there have also been practical instances where the iterates of QSA-A have been shown to be unstable (cf. \cite{prashanth2011reinforcement}). 

The problem is complicated due to the \emph{off-policy} nature of QSA-A. The \emph{off-policy} problem here arises because of the presence of the {\em min} operation in the Q-learning algorithm that introduces nonlinearity in the update rule (see term (I) in \eqref{eq:qsa-a-update-rule}). There is also a minor problem of estimating the average cost that involves a min operation as well (see term (II) in \eqref{eq:qsa-a-update-rule}). The latter problem can be solved by estimating the average cost in a separate recursion (as we do in \eqref{eq:tqsa-avg}) and using this estimate in place of the term (II). 

A nested two-loop procedure to overcome the off-policy problem works as follows:
\begin{description}
\item[{\em Inner loop.}] Instead of the first min operation (term (I) in \eqref{eq:qsa-a-update-rule}), employ a stochastic gradient technique to find the best action that minimizes the approximate Q-value function. A popular scheme for estimating the gradient of a function from simulation is SPSA and we employ a one-simulation SPSA scheme with deterministic perturbations for estimating $\nabla_w Q(s,a)$.
\item[{\em Outer loop.}] instead of the {\em min} operation, actions are selected according to a given policy, then the Q-learning update would resemble a temporal difference (TD) learning update for the joint (state-action) Markov chain. It has been shown in \cite{tsitroy1} that TD with linear function approximation converges. 
\end{description}
For ensuring convergence of the above procedure, one would have to run the two loops in a serial fashion for sufficiently long duration. This may be time-consuming and also result in slow convergence. To overcome this problem, we employ multi-timescale stochastic approximation \cite[Chapter 6]{borkar2008stochastic} to mimic the two-loop behavior, albeit with different step-sizes for the inner and outer loops. In other words, both the loops are allowed to run simultaneously, with a larger step-size for the inner loop and a smaller one for the outer loop. This achieves the effect of the nested loop procedure, while ensuring rapid convergence. 

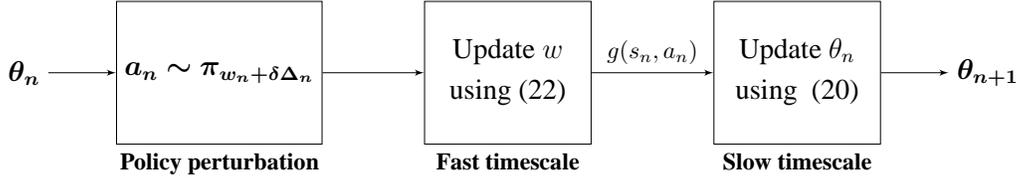
\begin{figure}
\centering
\tikzstyle{block} = [draw, fill=white, rectangle,
   minimum height=6em, minimum width=7em]
\tikzstyle{sum} = [draw, fill=white, circle, node distance=1cm]
\tikzstyle{input} = [coordinate]
\tikzstyle{output} = [coordinate]
\tikzstyle{pinstyle} = [pin edge={to-,thin,black}]
\scalebox{0.9}{\begin{tikzpicture}[auto, node distance=2cm,>=latex']
\node (theta) {\large $\bm{\theta_n}$};
\node [block, right=1cm of theta, label=below:{\bf Policy perturbation}] (psim) {\large$\bm{a_n \sim \pi_{w_n+\delta\Delta_n}}$}; 
\node [block, right=1.5cm of psim, label=below:{\bf Fast timescale}] (spsa) {\makecell[c]{\large Update $w$ \\[1ex] \large using \eqref{eq:tqsa-a-w-update}}}; 
\node [block, right=1.8cm of spsa, label=below:{\bf Slow timescale}] (update) {\makecell[c] {\large Update $\theta_n$ \\[1ex] \large using ~\eqref{eq:tqsa-a-update}}};
\node (thetanext)[right=1cm of update] {\large $\bm{\theta_{n+1}}$};
\draw [->] (theta) --  (psim);
\draw [->] (psim) --  (spsa);
\draw [->] (spsa) -- node {$g(s_n,a_n)$}  (update);
\draw [->] (update) --  (thetanext);
\end{tikzpicture}}
\caption{Overall flow of the TQSA-A algorithm.}
\label{fig:algorithm-flow}
\end{figure}

Recall that we consider a class of parameterized policies satisfying Assumption \ref{a1}\footnote{
One may use an  $\epsilon$-greedy policy for TQSA-A as well, however, that will result in additional exploration. Since TQSA-A updates the parameters of an underlying parameterized Boltzmann policy (which by itself is randomized in nature), we do not need an extra exploration step in our algorithm.}.
As illustrated in Fig. \ref{fig:algorithm-flow}, the idea in the gradient estimate is to simulate the system with the perturbed policy parameter $w + \delta \Delta$, where $\delta >0$ is a fixed small constant and $\Delta = (\Delta_1,\ldots,\Delta_N)^T$ are perturbations constructed using certain Hadamard matrices (see Lemma 3.3 of \cite{bhatnagar2003two} for details of the construction). Given the output from the perturbed simulation,  the gradient of the approximate Q-value function $Q(s,a) \approx \theta^T \sigma_{s,a}$ is estimated as:
\begin{equation}
\label{eq:spsa-estimate}
    \nabla_w Q(s,a) \approx \dfrac{\theta^{T}\sigma_{s,a}}{\delta}\Delta^{-1},
\end{equation}
where $\Delta^{-1} \stackrel{\triangle}{=} (\Delta_1^{-1},\ldots,\Delta_N^{-1})^T$.
 It has been shown in \cite{bhatnagar2003two} that an incremental stochastic recursive algorithm that incorporates the RHS of \eqref{eq:spsa-estimate} as its update direction essentially performs a search in the gradient direction when $\delta$ is small.

 The overall update of the TQSA-A proceeds on two different timescales as follows:\\
\begin{inparaenum}[\bfseries(i)]
    \item On the faster timescale, the policy parameter is updated along a gradient descent direction using an SPSA estimate \eqref{eq:spsa-estimate};\\
    \item On the slower timescale, the average cost \eqref{eq:averagecost} is estimated and
    \item Also, on the slower timescale, the Q-value parameter is updated in an on-policy TD algorithm-like fashion.\\
\end{inparaenum}
The update rule for the TQSA-A algorithm is given as follows: $\forall n \ge 0$,
\begin{align}
 \label{eq:tqsa-a-update}
 \theta_{n+1} = & \Gamma_1 \bigg( \theta_n + b(n)\sigma_{s_n,a_n}(g(s_n,a_n) - \hat J_{n+1}  +  \theta_{n}^{T}\sigma_{s_{n+1},a_{n+1}} -\theta_{n}^{T}\sigma_{s_n,a_n})\bigg),\\
  \hat J_{n+1} = & \hat J_n + c(n) \left( g(s_n,a_n) - \hat J_n\right),\label{eq:tqsa-avg}\\
\label{eq:tqsa-a-w-update}
  w_{n+1} = & \Gamma_2 \left ( w_n - a(n)\dfrac{\theta_n^{T}\sigma_{s_n,a_n}}{\delta}\Delta_{n}^{-1} \right ).
\end{align}
In the above,
the choice of features $\sigma_{s_n,a_n}$ is the same as in the algorithm, QSA-A and is described in Section \ref{sec:features}.
$\Gamma_{1}:\mathbb{R}^{d}\rightarrow D$, $\Gamma_{2}:\mathbb{R}^{N} \rightarrow C$ are certain projection operators that project the iterates $\theta_{n}$ and $w_n, n \ge 1$ to certain prescribed compact and convex subsets $D$ and $C$ of $\R^d$ and $\R^N$, respectively. The recursions \eqref{eq:tqsa-a-update} and \eqref{eq:tqsa-a-w-update} remain stable because of these projection operators, a crucial requirement for convergence of TQSA-A. 
The step-sizes $b(n),c(n),a(n)$ satisfy the following assumption:
\begin{assumption}    
\label{a2}
\begin{align*}
\sum_{n} a(n) =  \sum_{n} b(n) =\infty, \sum_n (a^2(n) + b^2(n)) < \infty, \lim\limits_{n \rightarrow \infty} \dfrac{b(n)}{a(n)} = 0.
\end{align*}
  Further, $c(n) = k a(n)$ for some $k>0$.
  \end{assumption}
While the first two conditions above are standard in stochastic approximation for step-sizes, the last condition, i.e., $\frac{b(n)}{a(n)} \rightarrow 0$ ensures the necessary timescale separation between policy and Q-value parameter updates. In particular, it guarantees that the policy parameter $w$ is updated on the faster timescale and average cost $\hat J$ and Q-value parameter $\theta$ are updated on the slower timescale. 

It turns out that because of the timescale difference, the recursion \eqref{eq:tqsa-a-w-update} converges almost surely to a set $w(\theta)$ that is a function of parameter $\theta$ and is seen to be a compact subset of $\mathbb{R}^N$. Further, the slower recursion \eqref{eq:tqsa-a-update} can be seen to track a differential inclusion and converges almost surely to a closed connected internally chain transitive invariant set of this differential inclusion. This claim is made precise by the convergence result in the following section.

\subsection{Convergence of TQSA-A}
We outline the proof of convergence of the TQSA-A algorithm, with the details being available in an appendix to this paper. In addition to assumptions \ref{a1} and \ref{a2}, we make the following assumption for the analysis.

\begin{assumption}
The Markov chain induced by any policy $w$ is irreducible and aperiodic.
\end{assumption}
The above ensures that each state gets visited an infinite number of times over an infinite time horizon and is standard in policy gradient RL algorithms.

The ODE approach is adopted for analyzing the convergence of $\theta$ and $w$ recursions \eqref{eq:tqsa-a-update}. In essence, the two-timescale stochastic approximation architecture employed in the TQSA-A algorithm allows
\begin{inparaenum}[(i)] \item the faster timescale analysis of the $w$-recursion in \eqref{eq:tqsa-a-update} assuming that the slower $\theta$-recursion is constant (quasi-static), and
\item the slower timescale analysis of the $\theta$-recursion in \eqref{eq:tqsa-a-update} assuming that the faster $w$-recursion has converged.
\end{inparaenum}
The convergence analysis comprises of the following important steps:
\begin{itemize}
    \item Theorem \ref{thm:main}, in effect, states that the $w$-recursion performs a gradient descent using one-simulation SPSA and converges to a set of points in the neighborhood of the local minimum of the approximate Q-value function $R(\theta, w)$ (defined below). Note that this analysis is for the $w$-recursion on the faster timescale, assuming the Q-value function parameter $\theta$ to be a constant.
    \item Analyzing the $\theta$-recursion on the slower timescale, Theorem \ref{prop1} claims that the iterate $\theta$ asymptotically converges to a closed connected internally chain transitive set associated with a corresponding differential inclusion (DI).
\end{itemize}
We present below the precise statements of these results.
Let ${\cal C}(C) ({\cal C}(D))$ denote the space of all continuous functions
from $C$  to ${\cal R}^N$ ($D$  to ${\cal R}^d$).
We define the operator $\hat{\Gamma}_2: {\cal C}(C) \rightarrow {\cal C}({\cal R}^N)$ as follows:
\begin{align*}
\hat{\Gamma}_2(v(w)) = &\lim_{\alpha \downarrow 0} \left(
\frac{\Gamma_2(w+\alpha v(w)) - w}{\alpha}\right).
\end{align*}
Consider the ODE associated with the $w$-recursion on the faster timescale, assuming $\theta(t)\equiv \theta$ (a constant independent of $t$):
\begin{equation}
\label{ode:d1}
\dot{w}(t) = \hat{\Gamma}_2\left(-\nabla_w R(\theta,w(t)) \right).
\end{equation}

Theorem \ref{thm:main} establishes that the $w$-recursion tracks the above ODE.
In the above,
\[R(\theta,w) \stackrel{\triangle}{=} \sum_{i\in S, a\in A(i)} f_{w}(i,a)\theta^T\sigma_{i,a},\]
where $f_{w}(i,a)$ are the stationary probabilities
$f_{w}(i,a)=d^{\pi_{w}}(i)\pi_{w}(i,a)$, $i\in S$, $a\in A(i)$ for the joint process $\{(X_n, Z_n)\}$, obtained from the state-action tuples at each instant. Here $d^{\pi_{w}}(i)$ is the stationary probability for the Markov chain $\{X_n\}$ under policy $\pi_w$ being in state $i\in S$. 
Let $K_\theta$ denote the set of asymptotically stable equilibria
of (\ref{ode:d1}), i.e., the local minima of the function
$R(\theta,\cdot)$ within the constraint set $C$.
Given $\epsilon >0$, let $K_\theta^\epsilon$ denote the $\epsilon$-neighborhood of $K_\theta$, i.e.,
\[ K_\theta^\epsilon = \{w\in C\mid \parallel w-w_0\parallel <\epsilon, w_0 \in K_\theta\}.\]

\begin{theorem}
\label{thm:main}
Let $\theta_n \equiv \theta, \forall n$, for some $\theta\in D\subset {\cal R}^d$. Then,
given $\epsilon > 0$, there exists $\delta_0 > 0$
such that for all $\delta \in (0,\delta_0]$,
$\{w_n\}$ governed by (\ref{eq:tqsa-a-update})
converges the set $K_\theta^{\epsilon}$ a.s.
\end{theorem}

We now analyze the $\theta$-recursion, which is the slower recursion in \eqref{eq:tqsa-a-update}.
Let \\$T_w:{\cal R}^{|S\times A(S)|}\rightarrow {\cal R}^{|S\times A(S)|}$ be the operator given by
\begin{equation}
\label{eq:T_a}
T_w(J)(i,a) =
g(i,a) - J(\pi_w) e
+ \sum_{j\in S, b\in A(j)}p_{w}(i,a; j,b) J(j,b),
\end{equation}
or in more compact notation
\[ T_w(J) = G - J(\pi_w) e + P_{w} J, \]
where $G$ is the column vector with components $g(i,a), i\in S, a\in A(i)$, $J(\pi_w)$ is the average cost corresponding to the policy parameter $w$ and $P_{w}$ is the transition probability matrix of the joint (state-action) Markov chain under policy $\pi_w$, with components $p_{w}(i,a; j,b)$. Here $p_{w}(i,a; j,b)$ denote the transition probabilities of the joint process $\{(X_n, Z_n)\}$.
The differential inclusion associated with the $\theta$-recursion of \eqref{eq:tqsa-a-update} corresponds to
\begin{equation}
\label{td-di-a}
\dot{\theta}(t) \in \hat{\Gamma}_\theta\left(h(\theta)\right),
\end{equation}
where $h(\theta)$ is the set-valued map, defined in compact notation as follows:
\[h(\theta) \stackrel{\triangle}{=} \{\Phi^T \gDash (T_{w(\theta)} (\Phi \theta)-\Phi \theta\mid w \in K^\epsilon_\theta\}.\]
In the above, $\F$ denotes
the diagonal matrix with elements along the diagonal being $f_{w}(i,a)$, $i\in S$, $a\in A(i)$.
Also, $\Phi$ denotes the matrix with rows $\sigma_{s,a}^T,~ s \in S, a \in A(s)$. The number of
rows of this matrix is thus $|S \times A(S)|$, while the number of columns is $N$.
Thus,
$\Phi = (\Phi(i), i=1,\ldots, N)$ where $\Phi(i)$ is the column vector
\[\Phi(i) = (\sigma_{s,a}(i), s \in S, a \in A(s))^T, ~ i=1,\ldots,N.\]
Further, the projection operator $\hat\Gamma_\theta$ is defined as
\[\hat\Gamma_\theta \stackrel{\triangle}{=} \cap_{\epsilon>0} ch\left( \cup_{\parallel \beta-\theta \parallel < \epsilon} \{ \gamma_1(\beta;y+Y) \mid y \in h(\beta), Y \in R(\beta) \} \right),\text{ where}\]

\begin{itemize}[$\bullet$]
\item $ch(S)$ denotes the closed convex hull of the set $S$;
\item $\gamma_1(\theta;y)$ denotes the directional derivative of $\Gamma_1$ at $\theta$ in the direction $y$ and is defined by
$$\gamma_1(\theta;y) \stackrel{\triangle}{=} \lim_{\eta\downarrow0} \left(\dfrac{\Gamma_1(\theta_n + \eta y)-\theta}{\eta} \right);$$
 \item $Y(n+1)$ is defined as follows:
\begin{align*}
Y(n+1) &\stackrel{\triangle}{=} 
\left(g(X_n,Z_n) - J(\pi_{w_n}) +
\theta_n^T\sigma_{X_{n+1},Z_{n+1}} - \theta_n^T\sigma_{X_n,Z_n}\right)\sigma_{X_n,Z_n}\\
 - &E\left[
\left(g(X_n,Z_n) - J(\pi_{w_n}) +
\theta_n^T\sigma_{X_{n+1},Z_{n+1}} - \theta_n^T\sigma_{X_n,Z_n}\right)\sigma_{X_n,Z_n} \mid {\cal G}(n)
\right],
\end{align*}
where ${\cal G}(n)$ $= \sigma (\theta_r, X_r, Z_r, r\leq n),n\geq 0$ is a sequence
of associated sigma fields; and
\item $R(\beta)$ denotes the compact support of the conditional distribution of $Y(n+1)$ given ${\cal G}(n)$.
\end{itemize}

The main result is then given as follows:
\begin{theorem}
\label{prop1}
The iterate $\theta_n$, $n\geq 0$ governed by (\ref{eq:tqsa-a-update}), converges a.s to a closed connected internally chain transitive invariant set of \eqref{td-di-a}.
\end{theorem}
The detailed proofs of Theorems \ref{thm:main} and \ref{prop1} are provided in the supplementary material.

\section{Intruder's Mobility Model Estimation}
\label{sec:unknownP}
The algorithms described in the previous sections assume knowledge of the transition dynamics (the matrix $P$) of the Markov chain governing the intruder movement. However, in practice, this information is not available. In this section, we present a procedure to estimate $P$ and combine the same with the sleep-wake scheduling algorithms described in the previous section. We assume that $P$ is stationary, i.e., it does not change with time.

The estimation procedure for $P$ is online and convergent. The combination with the sleep-wake scheduling algorithms happens via multi-timescale stochastic approximation. In essence, we run the estimation procedure for $P$ on the faster timescale while the updates for the parameters of the sleep-wake scheduling algorithms are conducted on the slower timescale.
Thus, the update recursions for the individual sleep-wake algorithms see the estimate for $P$ as equilibrated, i.e., converged.

Let $\hat P_0$ be the initial estimate of the transition probability matrix $P$. Then, the estimate $\hat P_n$ at time instant $n$ is tuned as follows:
\begin{equation}
\label{eqn:P-update}
\hat P_{n + 1} = \Pi\left ( \hat P_n + d(n) \hat{p}_n \hat{p}_{n + 1}^T \right ).
\end{equation}
In the above, $\hat{p}_n = \left [ p_n(i): i = 1, 2, \dots, N + 1 \right ]^T$ is a column vector signifying current location of the intruder. Further, $\Pi(\cdot)$ is a projection operator that ensures that the iterates $\hat P_n$ satisfy the properties of a transition probability matrix. Also,$\{d(n)\}$ is a step-size sequence chosen such that it is on the faster timescale, while the $\theta$-recursion of the algorithm described earlier is on the slower timescale.

The idea behind the above update rule can be explained as follows: 
Suppose the locations of the intruder at instants $n$ and $n + 1$ are known. Then, $\hat{p}_n$ and $\hat{p}_{n + 1}$ would be vectors with the value $1$ in $l_k$th position and $0$ elsewhere. The quantity $\hat{p}_n \hat{p}_{n + 1}^T$ would thus result in a matrix with $1$ at row index $l_k$ and column index $l_{k + 1}$ and $0$ elsewhere. The recursion \eqref{eqn:P-update} then results in a sample averaging behavior (due to stochastic approximation) for estimating the transition dynamics $P$. The same logic can be extended to the remaining cases, for instance, known $l_k$ and unknown $l_{k + 1}$ and so on.

Empirically we observe that the update \eqref{eqn:P-update} converges to the true transition probability matrix $P$ for each of the proposed algorithms, in all the network settings considered.

\section{Discounted Cost Setting}
\label{sec:discounted-objective}
We now describe the discounted cost objective. As in the case of the average cost setting (Section \ref{sec:average-setting}), we describe below the Bellman equation for continuous state-action spaces. However, for the sake of implementation (in later sections), we again use the discrete version of the problem.

For a policy $\pi$, define the value function $V^\pi:S\rightarrow \mathbb{R}$ as follows:
\begin{equation}
V^\pi(x) = E\left[\sum_{m=0}^{\infty} \gamma^m g(s_m,a_m)\mid X_0=x\right],
\label{vf}
\end{equation}
for all $x\in \S$. In the above, $\gamma \in (0,1)$ is a given discount factor.
The aim then is to find an optimal value function $V^*:\S\rightarrow \mathbb{R}$,
i.e.,
\begin{equation}
 V^*(x) = \min_{\pi\in \Pi} V^\pi(x) \stackrel{\triangle}{=} V^{\pi^*}(x),
\label{oe}
\end{equation}
where $\pi^*$ is the optimal policy, i.e., the one for which $V^*$ is the value function.
It is well known, see \cite{puterman},
that  the optimal value function $V^*(\cdot)$ satisfies
the following Bellman equation of optimality in the discounted cost case:
\begin{equation}
V^*(x) = \min_{a\in \A(x)} \left(g(x,a) + \gamma \int p(x,a,dy)V^*(y)\right),
\label{be}
\end{equation}
for all $x\in \S$. As for the average cost, our algorithms in the discounted cost setting do not require knowledge of the system model and incorporate function approximation.

\section{Discounted Cost Algorithms}
\label{sec:discounted-algorithms}
In this section, we present two algorithms for sleep-wake scheduling with the goal of minimizing a discounted cost objective described in Section \ref{sec:discounted-objective}. The overall structure of both the algorithms follow the schema provided in Algorithm \ref{algorithm:structure}. However, in comparison to the average cost algorithms described earlier, the parameter $\theta$ is updated in a different fashion here to cater to the discounted cost objective.  

\subsection{Q-learning based Sleep--wake Scheduling Algorithm (QSA-D)}
\label{sec:qsa-d}
As in the case of the average cost setting, the Q-learning algorithm cannot be used without employing function approximation because of the size of the state-action space. The function approximation variant of Q-learning in the discounted cost setting parameterizes the Q-values in a similar manner as the average cost setting, i.e., according to \eqref{eqn:Q_FA}. The algorithm works with a single online simulation trajectory
of states and actions, and updates $\theta$ according to
\begin{align}
\label{eq:qsa-d-update-rule}
\theta_{n+1} = \theta_n + a(n)\sigma_{s_n,a_n}\left (g(s_n,a_n) + \gamma \min_{b \in A(s_{n+1})}\theta_n^T\sigma_{s_{n+1},b} - \theta_n^T\sigma_{s_n,a_n}\right ),
\end{align}
where $\theta_0 $ is set arbitrarily.
In the above, $s_n$ and $s_{n+1}$
denote the state at instants $n$ and $n+1$, respectively, and $\theta_n$ denotes the $n^{th}$ update of the parameter. 
 In  (\ref{eq:qsa-d-update-rule}), the action $a_n$ is chosen in state $s_n$ according
to an $\epsilon-$greedy policy, as in the case of the QSA-A algorithm.

\subsection{Two-timescale Q-learning based sleep--wake scheduling algorithm (TQSA-D)}
\label{sec:tqsa-d}
As with the average cost setting, the Q-learning algorithm with function approximation in the discounted setting is not guaranteed to converge because of the off-policy problem.
A variant of Q-learning \cite{bhatnagar2012twotimescale} has been recently proposed and has been shown to be convergent. This algorithm uses two-timescale simultaneous perturbation 
stochastic approximation (SPSA) with Hadamard matrix based deterministic perturbation sequences \cite{bhatnagar2003two}.

 The TQSA-D algorithm is a two timescale stochastic approximation algorithm that employs a linear approximation architecture and parameterizes the policy. As in the case of TQSA-A, we assume here that the policy $\pi(s,a)$ is continuously
differentiable in the parameter $\theta$, for any state--action pair $(s,a)$. The function approximation parameter $\theta$ is tuned on the slower timescale in a TD-like fashion, while the policy parameter $w$ is tuned on the faster timescale in the negative gradient descent direction using SPSA. Let $\pi'_n \stackrel{\triangle}{=} \pi_{(w_{n} + \delta\Delta_{n})} = ( \pi_{(w_{n} + \delta\Delta_{n})}(i,a), i \in S, a \in A(i))^T$, where $\delta > 0$ is a given small constant, be the randomized policy parameterized by $(w_{n} + \delta\Delta_{n})$ during the $n$th instant. Here $\Delta_{n}, n\geq0$ are perturbations obtained from the Hadamard matrix based construction described before.
The update rule of the TQSA-D algorithm is given as follows: $\forall n \ge 0$,
\begin{align}
 \label{eq:tqsa-d-update-rule}
 \theta_{n+1} = \quad& \Gamma_{1}\left (\theta_n + b(n)\sigma_{s_n,a_n}\left(r(s_n,a_n) + \gamma \theta_{n}^{T}\sigma_{s_{n+1},a_{n+1}} -\theta_{n}^{T}\sigma_{s_n,a_n}\right) \right ),\nonumber \\ 
  w_{n+1} = \quad& \Gamma_{2} \left ( w_n - a(n)\dfrac{\theta_n^{T}\sigma_{s_n,a_n}}{\delta}\Delta_{n}^{-1} \right ).
\end{align}
The projection operators $\Gamma_1, \Gamma_2$ and the step-sizes $a(n),b(n)$ for all $n\ge0$ are the same as in TQSA-A and the features $\sigma_{s_n,a_n}$ are as in the previous algorithms.

\section{Simulation Setup and Results}
\label{sec:simulations}

\subsection{Implementation}
We implemented our sleep--wake scheduling algorithms - QSA-A and TQSA-A for the average cost setting and QSA-D and TQSA-D for the discounted cost setting, respectively. For the sake of comparison, we also implemented the FCR and \qmdp algorithms proposed in \cite{fuemmeler2008smart}. Note that for each of these algorithms, the knowledge of the mobility model of the intruder is assumed. We briefly recall these algorithms below:\\
\paragraph{FCR.}
This algorithm approximates the state evolution \eqref{pevolve} by
$p_{t+1}=p_{t}P$,
and then attempts to find the sleep time for each sensor by solving the following balance equation:
\begin{small}  
\begin{align*}
 V^{(l)}(p)=\min_{u} \big( \sum_{j=1}^{u} [pP^{j}]_{l} + \sum_{i=1}^{N}c[pP^{u+1}]_{i} + V^{(l)}(pP^{u+1}) \big).
\end{align*}
\end{small}
Thus, the sleeping policy here is obtained locally for each sensor by solving the above Bellman equation for each sensor, with a strong approximation on the state evolution. Note that our algorithms make no such assumptions and attempt to find the optimal sleeping policy in the global sense (i.e., considering all the sensors) and not in the local sense (i.e., treating the sensors individually).\\
\paragraph{\qmdp.}
In this approach, the decomposition into the per sensor problem is the same as in FCR. However here, the underlying assumption is that the location of the object will always be known in the future. Thus, instead of \eqref{pevolve}, the state evolves here according to
$p_{k+1}=e_{l_{k+1}}P$.
The objective function for a sensor $l$, given the state component $p$, is given by
\begin{small}
\begin{align*}
 V^{(l)}(p)=&\min_{u} \big ( \sum_{j=1}^{u} [pP^{j}]_{l} + \sum_{i=1}^{N}c[pP^{u+1}]_{i} + \sum_{i=1}^{N}[pP^{u+1}]_{i}V^{(l)}(e_{i}) \big ).
\end{align*}
\end{small}
The difference between the above and the corresponding equation for FCR is in the third term on the right hand side representing the future cost. In the case of $Q_{MDP}$, the future cost is the conditional expectation of the cost incurred from the object location after $u$ time units given the current distribution as its location. Thus, one can solve $V^{(l)}(p)$ for any $p$ once $V^{(l)}(e_{i}), 1 \leq i \leq N$ are known. The \qmdp algorithm then attempts to find a solution using the well-known dynamic programming procedure - policy iteration for MDPs.
However, an important drawback with the dynamic programming approaches is the curse of dimensionality (i.e., the computational complexity with solving the associated Markov decision process increases exponentially with the dimension and cardinality of the state and action spaces). RL algorithms that incorporate function approximation techniques alleviate this problem and make the computational complexity manageable, while still ensuring that these algorithms converge to a `good enough' policy.

\subsection{Simulation Setup}

We perform our experiments on a 2-D network setting (see Fig. \ref{fig:2d}) of $121$ sensors, i.e., a $11 \times 11$ grid. The sensor regions overlap here, with each sensor's sensing region overlapping with that of its neighboring nodes. In particular, the sensing regions of sensors in the interior of the grid overlap with eight neighboring nodes.

The simulations were conducted for $6000$ cycles for all algorithms.
We set the single-stage cost component $c$ to $0.1$ and the discount factor $\gamma$ to $0.9$.
For QSA-A/D, we set the exploration parameter $\epsilon$ to $0.1$. The projection operators $\Gamma_i, i=1,2$ are chosen such that each co-ordinate of $\theta$ and $w$ is forced to evolve within $[1,100]$.
The step-sizes are chosen as follows:
For QSA-A, we set
${\displaystyle a(n) = \frac{1}{n}, n \ge 1}$ and for TQSA-A, we set
$b(n)= \frac{1}{n}, \quad a(n)= \frac{1}{n^{0.55}}, n \ge 1$, respectively.
Further, for TQSA-A/TQSA-D, we set $\delta = 0.001$. 
For QSA-A, we choose the fixed state $s$ (see \eqref{eq:qsa-a-update-rule}) as $\langle p_{0},r \rangle$ where $p_{0}$ is the initial distribution of $p_k$ and $r$ is a random sleep time vector. It is easy to see that this choice ensures that there is a positive probability of the underlying MDP visiting state $s$\footnote{This is because the intruder stays in the starting location for at least one time step and the exploration of actions initially results in a positive probability of a random action being chosen.}. 

\subsection{Results}
We use the number of sensors awake and the number of detects per time step as the performance metrics for comparing the various sleep/wake algorithms. While the former metric is the ratio of the total number of sensors in the wake state to the number of time-steps, the latter is the ratio of the number of successful detects of the intruder to the number of time-steps.
Fig. \ref{average_tradeoff} presents the number of sensors awake and the number of detects per time step, for each of the algorithms studied in the average cost setting, while Fig. \ref{discounted_tradeoff} presents similar results for the algorithms in the discounted cost setting.

\begin{figure}
\centering
\begin{tabular}{c}
\subfloat[Number of detects per time step]{\hspace{-2em}\includegraphics[width=1.7in,height=2.4in,angle=270]{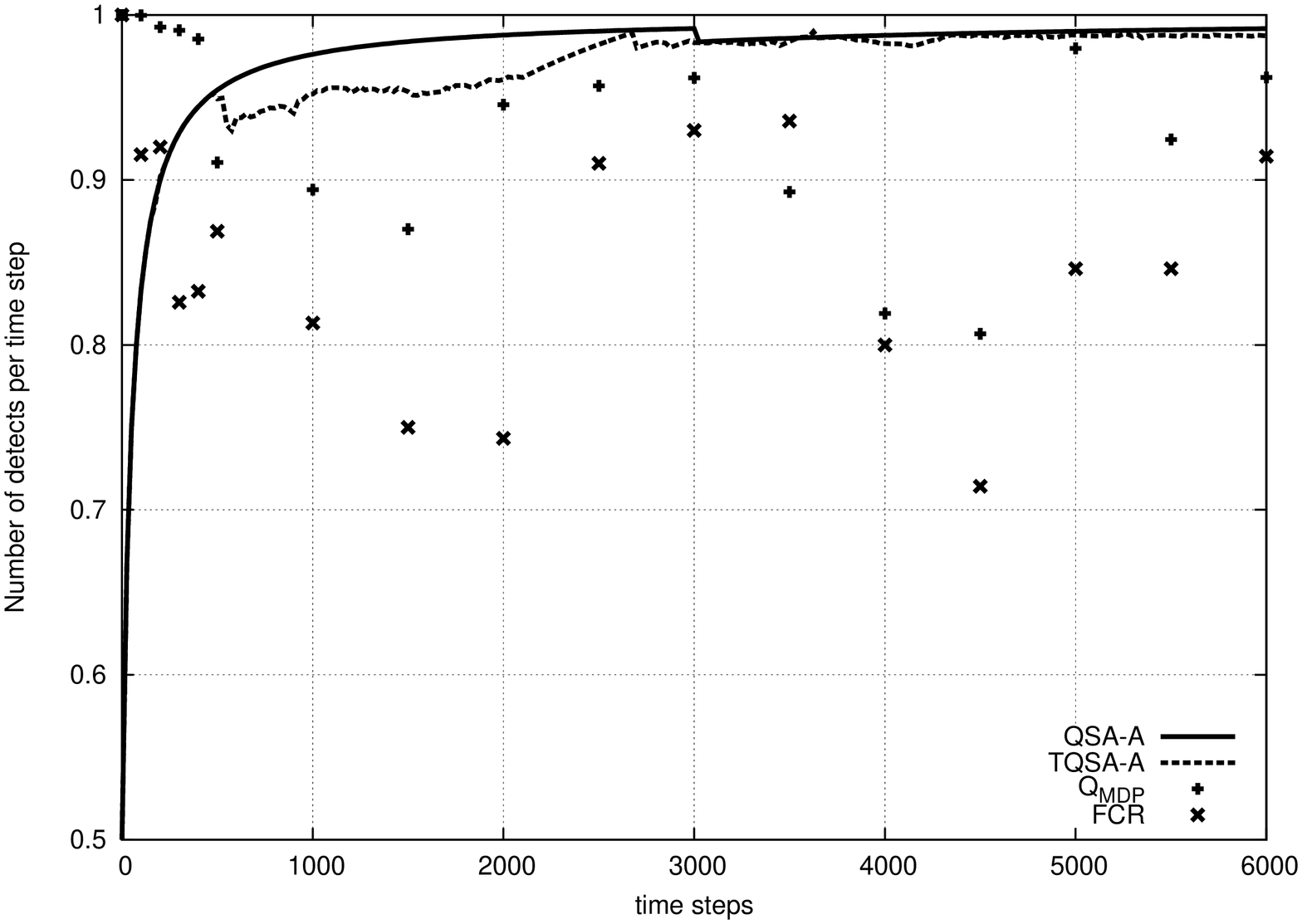}}
\subfloat[Number of sensors awake per time step]{\hspace{1em}\includegraphics[width=1.7in,height=2.4in,angle=270]{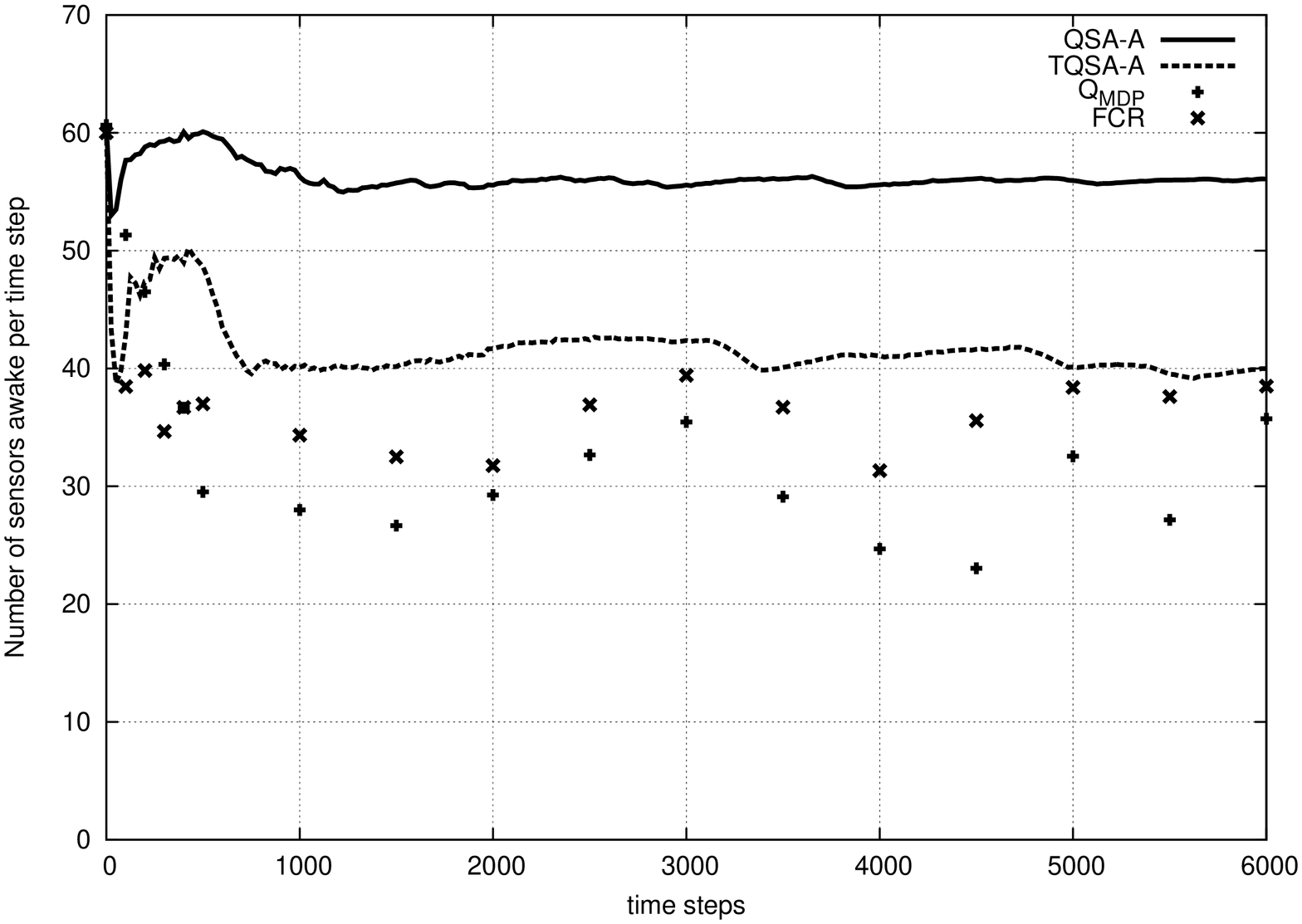}}
\end{tabular}
\caption{Tradeoff characteristics - known mobility model ($P$) case}
\label{average_tradeoff}
\end{figure}

Fig. \ref{a_qwfa_theta_values} presents the evolution of the Q-value parameter $\theta$ for TQSA-A in the average cost setting. 
Fig. \ref{a_2d_tradeoff_up} presents the results obtained from the experiments with TQSA-A combined with the mobility model estimation procedure \eqref{eqn:P-update}. Fig. \ref{a_P_values} shows the evolution of the estimate $P_{k}(i,j)$ of the intruder's mobility model, where $i$ corresponding to the $(6,6)$th cell and $j$ corresponding to $(6,5)$th cell, converges. In contrast, the \qmdp algorithm requires full knowledge of the distribution of the intruder movement and hence, cannot be applied in the setting of unknown $P$. 
 
\subsection{Discussion}
We observe that in comparison to the \qmdp algorithm, our algorithms attain a slightly higher tracking accuracy at the cost of a few additional sensors in the wake state. On the other hand, our algorithms exhibit better tradeoff between energy cost and tracking accuracy in comparison to the FCR algorithm. Amongst our algorithms, we observe that the two timescale variant TQSA-A  performs better than the Q-learning based QSA-A, since TQSA-A results in a tracking accuracy similar to QSA-A with lesser number of sensors awake. A similar observation holds in the discounted cost setting as well. 

Further, as evident in the tradeoff plot in Fig. \ref{average_tradeoff}, the \qmdp algorithm exhibits fluctuating behaviour with a significant number of outliers that show poor tradeoffs. This, we suspect, is due to the underlying requirement of complete future observations in \qmdp. Further, \qmdp (and even FCR) is not a learning algorithm that stabilizes the number of sensors awake and the tracking errors in the long-term. This is because, at each instant, \qmdp attempts to solve the  Bellman equation  in an approximate fashion and no information about the solution thus obtained is  carried forward to the future instants. 

On the contrary, our algorithms learn a good enough sleep/wake scheduling policy for the individual sensors with contextual information being carried forward from one time step to the next. This results in a stable regime for the number of sensors awake and the tracking accuracy, unlike \qmdp. 
While the number of sensors awake for the FCR algorithm is less than that for our algorithms, the tracking accuracy is significantly lower in comparison.  For critical tracking systems, where failing to track has higher penalty, our proposed algorithms (esp. TQSA-A) will be able to achieve greater performance (tracking accuracy) at the cost of only a few additional sensors in the wake state.

Further, it is evident from Fig. \ref{a_qwfa_theta_values} that the Q-value parameter $\theta$ of TQSA-A converges. This is a significant feature of the TQSA-A algorithm as it possesses theoretical convergence guarantees, unlike QSA-A, which may not converge in some settings. Moreover, it can also be seen that the transient period when the policy parameter $\theta$ has not converged, is short. It is worth noting here that providing theoretical rate of convergence results for TQSA-A is difficult. This is because rate results for multi-timescale stochastic approximation algorithms, except for those with linear recursions (see \cite{konda2004convergence}), is not known till date to the best of our knowledge.

We also observe that even for the case when the intruder's mobility model is not known, TQSA-A  shows performance on par with the vanilla TQSA-A, which assumes knowledge of $P$. 
We also observe that in the TQSA-A algorithm, the estimate $P_k$ of the transition probability matrix $P$ converges to the true $P$ and this is illustrated by the convergence plots in Fig. \ref{a_P_values}. 

\begin{figure}
\centering
    \begin{tabular}{c}
\subfloat[Number of sensors awake per time step]{\hspace{-2em}\includegraphics[width=1.7in,height=2.4in,angle=270]{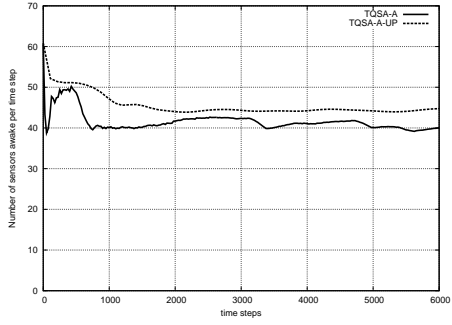}}
\subfloat[Number of detects per time step]{\hspace{1em}\includegraphics[width=1.7in,height=2.4in,angle=270]{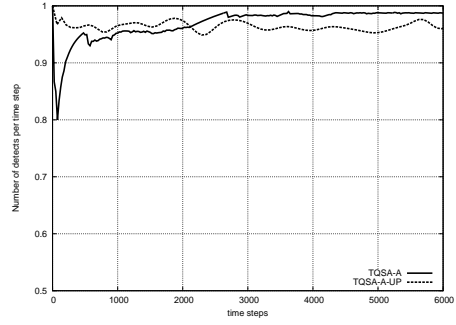}}
\end{tabular}
\caption{Tradeoff characteristics for TQSA-A algorithm with known and unknown $P$}
\label{a_2d_tradeoff_up}
\end{figure}

\begin{figure}
\centering
    \begin{tabular}{c}
\subfloat[Convergence of $\theta$]{\hspace{-2em}\includegraphics[width=1.7in,height=2.4in,angle=270]{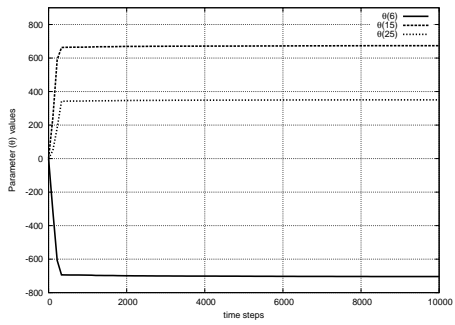}\label{a_qwfa_theta_values}}
\subfloat[Convergence of $P_k$]{\hspace{1em}\includegraphics[width=1.7in,height=2.4in,angle=270]{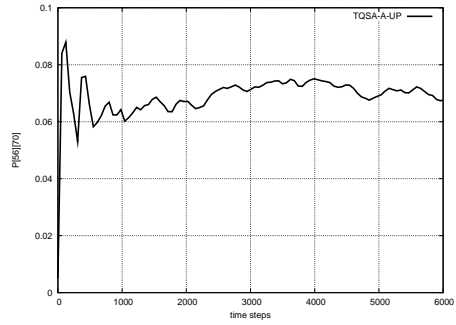}\label{a_P_values}}
\end{tabular}
\caption{Convergence trends of $\theta$ in TQSA-A with known $P$ and the estimate $P_k$ in TQSA-A with unknown $P$ }
\label{conv_trends}
\vspace{-4ex}
\end{figure}

\begin{figure}
\centering
\begin{tabular}{c}
\subfloat[Number of detects per time step]{\hspace{-2em}\includegraphics[width=1.7in,height=2.4in,angle=270]{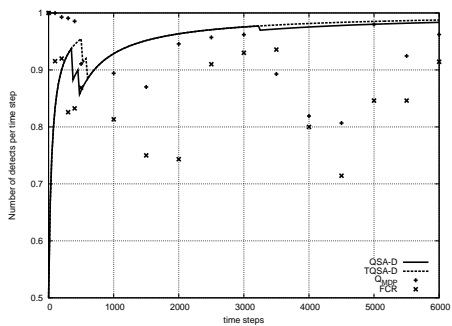}}
\subfloat[Number of sensors awake per time step]{\hspace{1em}\includegraphics[width=1.7in,height=2.4in,angle=270]{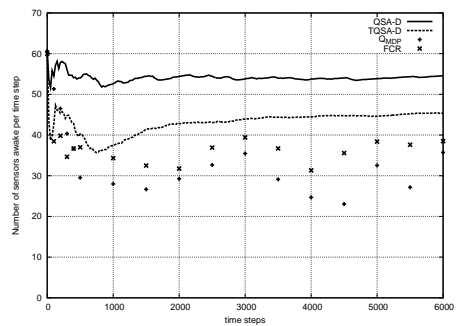}}
\end{tabular}
\caption{Tradeoff characteristics in the discounted setting}
\label{discounted_tradeoff}
\end{figure}

\section{Conclusions and Future Work}
\label{sec:conclusions}
We studied the problem of optimizing sleep times in a sensor network for intrusion detection. Following a POMDP formulation similar to the one in \cite{fuemmeler2011sleep}, our aim in this paper was to minimize certain long-run average and discounted cost objectives. This in turn allowed us to study both transient as well as steady state system behavior. For both the settings considered, we proposed a novel two-timescale Q-learning algorithm with theoretical convergence guarantees. For the sake of comparison, we also developed sleep-scheduling algorithms that are function approximation analogues of the well-known Q-learning algorithm. Next, we extended these algorithms to a setting where the intruder's mobility model is not known. Empirically, we demonstrated the usefulness of our algorithms on a simple two-dimensional network setting.

As future work, one could extend these algorithms to settings where multiple intruders have to be detected. This would involve the conflicting objectives of keeping less number of sensors awake and at the same time, detecting as many intruders as possible. Another interesting direction of future research is to develop intruder detection algorithms in a decentralized setting, i.e., a setting where the individual sensors collaborate in the absence of a central controller. Decentralized variants of our two-timescale Q-learning algorithm TQSA-A can be developed in the following manner: Each sensor runs TQSA-A to decide on the sleep times in a manner similar to the algorithms we propose. However, this would require the knowledge of $p_k$ (distribution of the intruder's location) at each sensor and this can be obtained by means of a message passing scheme between the individual sensors. Since exchanging messages between every pair of sensors may increase the load on the network, a practical alternative is to 
form (possibly dynamic) groups of sensors, within which the message regarding the intruder's location (or $p_k$) is exchanged. The individual sensors then decide on the sleep times using this local information and an update rule similar to \eqref{eq:tqsa-a-update}.

\appendix
\section{Convergence analysis for TQSA-A}
\label{sec:convergence-tqsa-a}

 The ODE approach is adopted for analyzing the convergence of $\theta$ and $w$ recursions \eqref{eq:tqsa-a-update} in the main paper. In essence, the two-timescale stochastic approximation architecture employed in the TQSA-A algorithm allows\\
\begin{inparaenum}[(i)] \item the faster timescale analysis of the $w$-recursion in \eqref{eq:tqsa-a-update} in the main paper assuming that the slower $\theta$-recursion is constant (quasi-static), and\\
    \item the slower timescale analysis of the $\theta$-recursion in \eqref{eq:tqsa-a-update} in the main paper assuming that the faster $w$-recursion has converged, for any given $\theta$.
\end{inparaenum}

The convergence analysis comprises of the following important steps:\\
\begin{inparaenum}[(i)]
    \item Theorem \ref{appendix:thm:main}, in effect, states that the $w$-recursion performs a gradient descent using one-simulation SPSA and converges to a set of points in the neighborhood of the local minimum of the approximate Q-value function $R(\theta, w)$ (defined below). Note that this analysis is for the $w$-recursion on the faster timescale, assuming the Q-value function parameter $\theta$ to be a constant.\\
    \item Analyzing the $\theta$-recursion on the slower timescale, Theorem \ref{appendix:prop1} claims that the iterate $\theta$ asymptotically converges to a closed connected internally chain transitive invariant set associated with a corresponding differential inclusion (DI).
\end{inparaenum}

\subsection{Analysis of the $w$-recursion}
We present below the precise statements of these results.
Let ${\cal C}(C) ({\cal C}(D))$ denote the space of all continuous functions
from $C$  to ${\cal R}^N$ ($D$  to ${\cal R}^d$).
We define the operator $\hat{\Gamma}_2: {\cal C}(C) \rightarrow {\cal C}({\cal R}^N)$ as follows:
\begin{align*}
\hat{\Gamma}_2(v(w)) = &\lim_{\alpha \downarrow 0} \left(
\frac{\Gamma_2(w+\alpha v(w)) - w}{\alpha}\right).
\end{align*}
Consider the ODE associated with the $w$-recursion on the faster timescale, assuming $\theta(t)\equiv \theta$ (a constant independent of $t$):
\begin{equation}
\label{appendix:ode:d1}
\dot{w}(t) = \hat{\Gamma}_2\left(-\nabla_w R(\theta,w(t)) \right),
\end{equation}
with
\[R(\theta,w) \stackrel{\triangle}{=} \sum_{i\in S, a\in A(i)} f_{w}(i,a)\theta^T\sigma_{i,a},\]
where $f_{w}(i,a)$ are the stationary probabilities
$f_{w}(i,a)=d^{\pi_{w}}(i)\pi_{w}(i,a)$, $i\in S$, $a\in A(i)$ for the joint process $\{(X_n, Z_n)\}$, obtained from the state-action tuples at each instant. Here $d^{\pi_{w}}(i)$ is the stationary probability distribution for the Markov chain $\{X_n\}$ under policy $\pi_w$ being in state $i\in S$. The ergodicity of the joint process $\{(X_n, Z_n)\}$ and the existence of stationary distribution $f_{w}(i,a)$ follows from the proposition below:

\begin{proposition}
\label{lem1}
Under (A1) and (A2), the
process $(X_n, Z_n), n\geq 0$ with $Z_n$, $n\geq 0$ obtained
from the SRP $\pi_w$, for any given $w\in C$,
is an ergodic Markov process.
\end{proposition}
\begin{proof}
    See \cite[Proposition 1, Section 3]{bhatnagar2012twotimescale}.
\end{proof}

\begin{lemma}
\label{lemma1}
Under (A1) and (A2), the stationary probabilities
$f_w(i,a)$, $i\in S$, $a\in A(i)$ are continuously
differentiable in the parameter $w\in C$.
\end{lemma}
\begin{proof}
    See \cite[Lemma 1, Section 3]{bhatnagar2012twotimescale}.
\end{proof}

Let $K_\theta$ denote the set of asymptotically stable equilibria
of (\ref{appendix:ode:d1}), i.e., the local minima of the function
$R(\theta,\cdot)$ within the constraint set $C$.
Given $\epsilon >0$, let $K_\theta^\epsilon$ denote the closed $\epsilon$-neighborhood of $K_\theta$, i.e.,
\[ K_\theta^\epsilon = \{w\in C\mid \parallel w-w_0\parallel \le\epsilon, w_0 \in K_\theta\}.\]

The following result establishes that the $w$-recursion tracks the ODE \eqref{appendix:ode:d1}.

\begin{theorem}
\label{appendix:thm:main}
Let $\theta_n \equiv \theta, \forall n$, for some $\theta\in D\subset {\cal R}^d$. Then,
given $\epsilon > 0$, there exists $\delta_0 > 0$
such that for all $\delta \in (0,\delta_0]$,
$\{w_n\}$ governed by (\ref{eq:tqsa-a-update}) in the main paper
converges to the set $K_\theta^{\epsilon}$ almost surely.
\end{theorem}
\begin{proof}
    See \cite[Theorem 2, Section 3]{bhatnagar2012twotimescale}.
\end{proof}

\begin{proposition}
\label{lemma-wcont}
The set $K_\theta^{\epsilon}$ is a compact subset of ${\cal R}^N$ for any $\theta$ and $\epsilon>0$.
\end{proposition}
\begin{proof}
Follows in a similar manner as  \cite[Corollary 2]{bhatnagar2012twotimescale}.
\end{proof}

\subsection{Analysis of the $\theta$-recursion}
We now analyze the $\theta$-recursion, which is the slower recursion in \eqref{eq:tqsa-a-update} in the main paper. We first show that the estimate $\hat J_n$ tracks the average cost $J(\pi_{w_n})$ corresponding to the policy parameter $w_n \equiv w(\theta_n)$.

\begin{lemma}
\label{thresholdtuning:lemma-hirsch-faster}
With probability one, $|\hat J_n-J(\pi_{w_n})|\rightarrow 0$ as $n\rightarrow\infty$, where $J(\pi_{w_n})$
is the average reward under $\pi_{w_n}$.
\end{lemma}

\begin{proof}
The $\hat{J}$ update can be re-written as
\begin{align}
\hat{J}_{n + 1} = \hat{J}_n + b(n)\left ( J(\pi_{w_n}) + \xi_n - \hat{J}_n + M_{n + 1} \right ).
\label{eq:hatJ}
\end{align}
In the above,
\begin{itemize}[$\bullet$]
    \item $\mathcal{F}_n = \sigma(w_m, \theta_n,\xi_m,M_m; m \le n), n \ge 0$ is a set of $\sigma$-fields.
    \item $\xi_n = (E[ g(s_n,a_n) | \mathcal{F}_{n - 1} ]  - J(\pi_{w_n})), n \ge 0$ and
    \item $M_{n + 1} = g(s_n,a_n) - E[ g(s_n,a_n) | \mathcal{F}_{n - 1} ], n \ge 0$ is a martingale difference sequence.
\end{itemize}
 Let $N_m = \sum_{n = 0}^{m} c(n) M_{n + 1}$. Clearly, $(N_m, \mathcal{F}_m), m \ge 0$ is a square-integrable and almost surely convergent martingale.
Further, from Proposition \ref{lem1},
$|\xi_n| \rightarrow 0$ almost surely on the `natural timescale', as $n \rightarrow \infty$.  The `natural timescale' is faster than the algorithm's timescale and hence $\xi_n$ vanishes asymptotically, almost surely, see \cite[Chapter 6.2]{borkar2008stochastic} for detailed treatment of natural timescale algorithms.

The ODE associated with \eqref{eq:hatJ} is
\begin{align}
    \dot{\hat J}(t) = J(\pi_{w(t)}) - \hat J(t) \stackrel{\triangle}{=} H(\hat J(t)).
\label{eq:ode-hatJ}
\end{align}
Let $H_{\infty}(\hat J(t)) = \lim_{c\rightarrow\infty} \dfrac{H(c\hat J(t))}{c} = -\hat J(t)$. Note that the ODE
\[ \dot{\hat J}(t) = -\hat J(t)\]
has the origin as its unique globally asymptotically stable equilibrium. Further, the ODE \eqref{eq:ode-hatJ} has ${\hat J}^* = J(\pi_{w_n})$ as its unique asymptotically stable equilibrium. The claim follows from Lemma 7 - Corollary 8 on pp. 74 and Theorem 9 on pp. 75 of \cite{borkar2008stochastic}.
\end{proof}

Let $T_w:{\cal R}^{|S\times A(S)|}\rightarrow {\cal R}^{|S\times A(S)|}$ be the operator given by
\begin{equation}
\label{appendix:eq:T_a}
T_w(J)(i,a) =
g(i,a) - J(\pi_w) e
+ \sum_{j\in S, b\in A(j)}p_{w}(i,a; j,b) J(j,b),
\end{equation}
or in more compact notation
\[ T_w(J) = G - J(\pi_w) e + P_{w} J, \]
where $G$ is the column vector with components $g(i,a), i\in S, a\in A(i)$.
$P_{w}$ is the transition probability matrix of the joint (state-action) Markov chain $\{(X_n, Z_n)\}$ under policy $\pi_w$, with components $p_{w}(i,a; j,b)$. Here $p_{w}(i,a; j,b)$ denote the transition probabilities of the joint process $\{(X_n, Z_n)\}$.

The differential inclusion associated with the $\theta$-recursion of \eqref{eq:tqsa-a-update} in the main paper corresponds to
\begin{equation}
\label{appendix:td-di-a}
\dot{\theta}(t) \in \hat{\Gamma}_\theta\left(h(\theta)\right),
\end{equation}
where $h(\theta)$ is the set-valued map, defined in compact notation as follows:
\[h(\theta) \stackrel{\triangle}{=} \{\Phi^T \gDash (T_{w(\theta)} (\Phi \theta)-\Phi \theta\mid w \in K^\epsilon_\theta\}.\]
In the above, $\F$ denotes
the diagonal matrix with elements along the diagonal being $f_{w}(i,a)$, $i\in S$, $a\in A(i)$.
Also, $\Phi$ denotes the matrix with rows $\sigma_{s,a}^T,~ s \in S, a \in A(s)$. The number of
rows of this matrix is thus $|S \times A(S)|$, while the number of columns is $N$.
Thus,
$\Phi = (\Phi(i), i=1,\ldots, N)$ where $\Phi(i)$ is the column vector
\[\Phi(i) = (\sigma_{s,a}(i), s \in S, a \in A(s))^T, ~ i=1,\ldots,N.\]
Further, the projection operator $\hat\Gamma_\theta$ is defined as
\[\hat\Gamma_\theta \stackrel{\triangle}{=} \cap_{\epsilon>0} ch\left( \cup_{\{\parallel \beta-\theta \parallel < \epsilon\}} \{ \gamma_1(\beta;y+Y) \mid y \in h(\beta), Y \in R(\beta) \} \right),\text{ where}\]

\begin{itemize}[$\bullet$]
\item $ch(S)$ denotes the closed convex hull of the set $S$;
\item $\gamma_1(\theta;y)$ denotes the directional derivative of $\Gamma_1$ at $\theta$ in the direction $y$ and is defined by
\begin{align}
\gamma_1(\theta;y) \stackrel{\triangle}{=} \lim_{\eta\downarrow0} \left(\dfrac{\Gamma_1(\theta_n + \eta y)-\theta}{\eta} \right); 
\label{eq:gamma1}
\end{align}
 \item $Y(n+1)$ is defined as follows:
\begin{align*}
Y(n+1) &\stackrel{\triangle}{=} 
\left(g(X_n,Z_n) - J(\pi_{w_n}) +
\theta_n^T\sigma_{X_{n+1},Z_{n+1}} - \theta_n^T\sigma_{X_n,Z_n}\right)\sigma_{X_n,Z_n}\\
 - &E\left[
\left(g(X_n,Z_n) - J(\pi_{w_n}) +
\theta_n^T\sigma_{X_{n+1},Z_{n+1}} - \theta_n^T\sigma_{X_n,Z_n}\right)\sigma_{X_n,Z_n} \mid {\cal G}(n)
\right],
\end{align*}
where ${\cal G}(n)$ $= \sigma (\theta_r, X_r, Z_r, r\leq n),n\geq 0$ is a sequence
of associated sigma fields; and
\item $R(\beta)$ denotes the compact support of the conditional distribution of $Y(n+1)$ given ${\cal G}(n)$.
\end{itemize}

The main result is then given as follows:
\begin{theorem}
\label{appendix:prop1}
The iterate $\theta_n$, $n\geq 0$ governed by (\ref{eq:tqsa-a-update}) in the main paper, converges a.s to a closed connected internally chain transitive invariant set of \eqref{appendix:td-di-a}.
\end{theorem}
\begin{proof}
Let $N(n) = \sum_{m = 0}^{n-1} b(m) Y(m + 1)$. It is easy to see that $N(n), n\ge0$ is a martingale sequence. Further, $(N(n), \mathcal{G}(n)), n \ge 0$ is a square-integrable and almost surely convergent martingale, owing to the following facts:\\ 
\begin{inparaenum}[\bfseries(i)]  
 \item $\sup\limits_{(i,a)\in S\times A(S)} \parallel \sigma_{i,a} \parallel < \infty$ and $\sup\limits_{(i,a)\in S\times A(S)} |g(i,a)| < \infty$ since $S\times A(S)$ is a finite set.\\
 \item Since we project the iterate $\theta$ using $\Gamma_1$ onto a compact and convex set $C$, we have $\sup\limits_n \parallel \theta_n \parallel < \infty$.\\
 \item By assumption, the step-size sequence $b(n),n\ge0$ satisfies $\sum\limits_n b(n)^2 < \infty$. \\
\end{inparaenum}
Thus, the $\theta$-recursion \eqref{eq:tqsa-a-update} in the main paper can be re-written as follows:
\begin{align}
\theta_{n+1} = & \Gamma_1\bigg(\theta_n + b(n) y_n + b(n)Y(n+1)\bigg),
\end{align}
Following the technique in \cite[Chapter 5.4]{borkar2008stochastic}, one can rewrite the above as follows:
 \begin{align}
\theta_{n+1} = & \theta_n + b(n)\left(\dfrac{\Gamma_1(\theta_n + b(n)(y_n + Y(n+1)))-\theta_n}{b(n)}\right),\nonumber\\
             = & \theta_n + b(n)\left(\gamma_1(\theta_n;y_n + Y(n+1)) + o(b(n))\right),\label{eq:t2}
\end{align}
where $\gamma_1(\theta;y)$ is as defined in \eqref{eq:gamma1}.
Let $z_n \stackrel{\triangle}{=} E[\gamma_1(\theta_n;y_n + Y(n+1))\mid {\cal G}(n)]$ and $\check Y(n+1) \stackrel{\triangle}{=} \gamma_1(\theta_n;y_n + Y(n+1)) - z_n$. Then, it is easy to see that \eqref{eq:t2} is equivalent to
 \begin{align}
\theta_{n+1} = & \theta_n + b(n)\left( z_n + \check Y(n+1) + o(b(n))\right).\label{eq:t3}
\end{align}

Using similar arguments as in  \cite[Proposition 3]{bhatnagar2012twotimescale}, it can be seen that $\hat\Gamma(h(\theta))$ satisfies the conditions stipulated in \cite[Section 5.1]{borkar2008stochastic}. These conditions ensure that  $\hat\Gamma(h(\theta))$ is compact, convex valued and upper-semicontinuous with bounded range on compacts.

Since we use $\Gamma_1$ operator to ensure $\theta$ is bounded, the general result of \cite[Corollary 4, Chapter 5]{borkar2008stochastic} can be applied to see that $\theta_n$ converges a.s. to a closed connected internally chain transitive invariant set of \eqref{appendix:td-di-a}.  
 The claim follows.
\end{proof}

\bibliography{rl_sleepwake}

\begin{thebibliography}{33}
\providecommand{\natexlab}[1]{#1}
\providecommand{\url}[1]{\texttt{#1}}
\expandafter\ifx\csname urlstyle\endcsname\relax
  \providecommand{\doi}[1]{doi: #1}\else
  \providecommand{\doi}{doi: \begingroup \urlstyle{rm}\Url}\fi

\bibitem[Abounadi et~al.(2002)Abounadi, Bertsekas, and
  Borkar]{abounadi2002learning}
J.~Abounadi, D.~Bertsekas, and V.S. Borkar.
\newblock {Learning algorithms for Markov decision processes with average
  cost}.
\newblock \emph{SIAM Journal on Control and Optimization}, 40\penalty0
  (3):\penalty0 681--698, 2002.
\newblock ISSN 0363-0129.

\bibitem[Baird(1995)]{baird}
L.~Baird.
\newblock Residual algorithms: Reinforcement learning with function
  approximation.
\newblock In \emph{ICML}, pages 30--37, 1995.

\bibitem[Beccuti et~al.(2009)Beccuti, Codetta-Raiteri, and
  Franceschinis]{beccuti2009multiple}
M.~Beccuti, D.~Codetta-Raiteri, and G.~Franceschinis.
\newblock Multiple abstraction levels in performance analysis of wsn monitoring
  systems.
\newblock In \emph{International ICST Conference on Performance Evaluation
  Methodologies and Tools}, page~73, 2009.

\bibitem[Bertsekas(2007)]{BertsekasDP01}
Dimitri~P. Bertsekas.
\newblock \emph{Dynamic Programming and Optimal Control, vol. II, 3rd edition}.
\newblock Athena Scientific, 2007.

\bibitem[Bertsekas and Tsitsiklis(1996)]{BertsekasT96}
Dimitri~P. Bertsekas and John~N. Tsitsiklis.
\newblock \emph{Neuro-Dynamic Programming}.
\newblock {Athena Scientific}, May 1996.
\newblock ISBN 1886529108.

\bibitem[Bhatnagar and Lakshmanan(2012)]{bhatnagar2012twotimescale}
S.~Bhatnagar and K.~Lakshmanan.
\newblock {A New Q-learning Algorithm with Linear Function Approximation}.
\newblock Technical report, SSL, IISc, 2012.
\newblock URL
  \url{http://stochastic.csa.iisc.ernet.in/www/research/files/IISc-CSA-SSL-TR-%
2012-3.pdf}.

\bibitem[Bhatnagar et~al.(2003)Bhatnagar, Fu, Marcus, and
  Wang]{bhatnagar2003two}
S.~Bhatnagar, M.C. Fu, S.I. Marcus, and I.~Wang.
\newblock {Two-timescale simultaneous perturbation stochastic approximation
  using deterministic perturbation sequences}.
\newblock \emph{ACM Transactions on Modeling and Computer Simulation (TOMACS)},
  13\penalty0 (2):\penalty0 180--209, 2003.
\newblock ISSN 1049-3301.

\bibitem[Bhatnagar et~al.(2013)Bhatnagar, Prasad, and Prashanth]{Bhatnagar13SR}
S.~Bhatnagar, H.~Prasad, and {L.A.} Prashanth.
\newblock \emph{Stochastic Recursive Algorithms for Optimization}, volume 434.
\newblock Springer, 2013.

\bibitem[Bhatnagar et~al.(2009)Bhatnagar, Sutton, Ghavamzadeh, and
  Lee]{bhatnagar2009natural}
Shalabh Bhatnagar, Richard~S Sutton, Mohammad Ghavamzadeh, and Mark Lee.
\newblock Natural actor--critic algorithms.
\newblock \emph{Automatica}, 45\penalty0 (11):\penalty0 2471--2482, 2009.

\bibitem[Borkar(2008)]{borkar2008stochastic}
V.S. Borkar.
\newblock \emph{{Stochastic Approximation: A Dynamical Systems Viewpoint}}.
\newblock Cambridge Univ Press, 2008.

\bibitem[Cui et~al.(2012{\natexlab{a}})Cui, Lau, Wang, Huang, and
  Zhang]{cui2012survey}
Ying Cui, Vincent~KN Lau, Rui Wang, Huang Huang, and Shunqing Zhang.
\newblock {A Survey on Delay-Aware Resource Control for Wireless
  Systems—Large Deviation Theory, Stochastic Lyapunov Drift, and Distributed
  Stochastic Learning}.
\newblock \emph{IEEE Transactions on Information Theory}, 58\penalty0
  (3):\penalty0 1677--1701, 2012{\natexlab{a}}.

\bibitem[Cui et~al.(2012{\natexlab{b}})Cui, Lau, and Wu]{cui2012delay}
Ying Cui, Vincent~KN Lau, and Yueping Wu.
\newblock {Delay-aware BS discontinuous transmission control and user
  scheduling for energy harvesting downlink coordinated MIMO systems}.
\newblock \emph{IEEE Transactions on Signal Processing}, 60\penalty0
  (7):\penalty0 3786--3795, 2012{\natexlab{b}}.

\bibitem[Fu and van~der Schaar(2009)]{fu2009learning}
Fangwen Fu and Mihaela van~der Schaar.
\newblock Learning to compete for resources in wireless stochastic games.
\newblock \emph{IEEE Transactions on Vehicular Technology}, 58\penalty0
  (4):\penalty0 1904--1919, 2009.

\bibitem[Fuemmeler and Veeravalli(2008)]{fuemmeler2008smart}
J.A. Fuemmeler and V.V. Veeravalli.
\newblock Smart sleeping policies for energy efficient tracking in sensor
  networks.
\newblock \emph{IEEE Transactions on Signal Processing}, 56\penalty0
  (5):\penalty0 2091--2101, 2008.

\bibitem[Fuemmeler et~al.(2011)Fuemmeler, Atia, and
  Veeravalli]{fuemmeler2011sleep}
J.A. Fuemmeler, G.K. Atia, and V.V. Veeravalli.
\newblock Sleep control for tracking in sensor networks.
\newblock \emph{IEEE Transactions on Signal Processing}, 59\penalty0
  (9):\penalty0 4354--4366, 2011.

\bibitem[Gui and Mohapatra(2004)]{gui2004power}
C.~Gui and P.~Mohapatra.
\newblock Power conservation and quality of surveillance in target tracking
  sensor networks.
\newblock In \emph{Proceedings of the international conference on mobile
  computing and networking}, pages 129--143, 2004.

\bibitem[Jiang et~al.(2008)Jiang, Han, Ravindran, and Cho]{jiang2008energy}
B.~Jiang, K.~Han, B.~Ravindran, and H.~Cho.
\newblock Energy efficient sleep scheduling based on moving directions in
  target tracking sensor network.
\newblock In \emph{IEEE International Symposium on Parallel and Distributed
  Processing}, pages 1--10, 2008.

\bibitem[Jianlin et~al.(2009)Jianlin, Fenghong, and Hua]{jianlin2009rl}
Mao Jianlin, Xiang Fenghong, and Lai Hua.
\newblock {RL-based superframe order adaptation algorithm for IEEE 802.15.4
  networks}.
\newblock In \emph{Chinese Control and Decision Conference}, pages 4708--4711.
  IEEE, 2009.

\bibitem[Jin et~al.(2006)Jin, Lu, and Park]{jin2006dynamic}
Guang-yao Jin, Xiao-yi Lu, and Myong-Soon Park.
\newblock {Dynamic Clustering for Object Tracking in Wireless Sensor Networks}.
\newblock \emph{Ubiquitous Computing Systems}, 4239:\penalty0 200--209, 2006.

\bibitem[Khan and Rinner(2012)]{khan2012resource}
Muhidul~Islam Khan and Bernhard Rinner.
\newblock Resource coordination in wireless sensor networks by cooperative
  reinforcement learning.
\newblock In \emph{IEEE International Conference on Pervasive Computing and
  Communications Workshop}, pages 895--900. IEEE, 2012.

\bibitem[Konda and Tsitsiklis(2004)]{konda2004convergence}
Vijay~R Konda and John~N Tsitsiklis.
\newblock Convergence rate of linear two-time-scale stochastic approximation.
\newblock \emph{Annals of Applied Probability}, pages 796--819, 2004.

\bibitem[Liu and Elhanany(2006)]{liu2006rl}
Z.~Liu and I.~Elhanany.
\newblock {RL-MAC: a QoS-aware reinforcement learning based MAC protocol for
  wireless sensor networks}.
\newblock In \emph{IEEE International Conference on Networking, Sensing and
  Control}, pages 768--773. IEEE, 2006.

\bibitem[Niu(2010)]{niu2010self}
Jianjun Niu.
\newblock Self-learning scheduling approach for wireless sensor network.
\newblock In \emph{International Conference on Future Computer and
  Communication (ICFCC)}, volume~3, pages 253--257. IEEE, 2010.

\bibitem[Prashanth and Bhatnagar(2011{\natexlab{a}})]{la2011reinforcement}
L.~A. Prashanth and S.~Bhatnagar.
\newblock Reinforcement learning with function approximation for traffic signal
  control.
\newblock \emph{IEEE Transactions on Intelligent Transportation Systems},
  12\penalty0 (2):\penalty0 412 -- 421, 2011{\natexlab{a}}.

\bibitem[Prashanth and
  Bhatnagar(2011{\natexlab{b}})]{prashanth2011reinforcement}
L.~A. Prashanth and S.~Bhatnagar.
\newblock Reinforcement learning with average cost for adaptive control of
  traffic lights at intersections.
\newblock In \emph{14th International IEEE Conference on Intelligent
  Transportation Systems (ITSC)}, pages 1640--1645. IEEE, 2011{\natexlab{b}}.

\bibitem[Prashanth et~al.(2014)Prashanth, Chatterjee, and Bhatnagar]{comsnets}
L.A. Prashanth, Abhranil Chatterjee, and Shalabh Bhatnagar.
\newblock Adaptive sleep-wake control using reinforcement learning in sensor
  networks.
\newblock In \emph{Sixth International Conference on Communication Systems and
  Networks (COMSNETS)}. IEEE, 2014.

\bibitem[Premkumar and Kumar(2008)]{premkumar2008optimal}
K.~Premkumar and A.~Kumar.
\newblock Optimal sleep--wake scheduling for quickest intrusion detection using
  sensor networks.
\newblock \emph{IEEE INFOCOM, Arizona, USA}, 2008.

\bibitem[Puterman(1994)]{puterman}
M.L. Puterman.
\newblock \emph{{Markov decision processes: Discrete stochastic dynamic
  programming}}.
\newblock John Wiley \& Sons, Inc. New York, NY, USA, 1994.

\bibitem[Rucco et~al.(2013)Rucco, Bonarini, Brandolese, and
  Fornaciari]{rucco2013bird}
Luigi Rucco, Andrea Bonarini, Carlo Brandolese, and William Fornaciari.
\newblock {A bird's eye view on reinforcement learning approaches for power
  management in WSNs}.
\newblock In \emph{Wireless and Mobile Networking Conference (WMNC)}, pages
  1--8. IEEE, 2013.

\bibitem[Spall(1992)]{spall1992multivariate}
James~C Spall.
\newblock Multivariate stochastic approximation using a simultaneous
  perturbation gradient approximation.
\newblock \emph{IEEE Transactions on Automatic Control}, 37\penalty0
  (3):\penalty0 332--341, 1992.

\bibitem[Sutton and Barto(1998)]{sutton1998reinforcement}
R.S. Sutton and A.G. Barto.
\newblock \emph{Reinforcement learning: An introduction}.
\newblock Cambridge Univ Press, 1998.

\bibitem[Tsitsiklis and Van~Roy(1997)]{tsitroy1}
John~N Tsitsiklis and Benjamin Van~Roy.
\newblock {An Analysis of Temporal Difference Learning with Function
  Approximation}.
\newblock \emph{IEEE Transactions on Automatic Control}, 42\penalty0
  (5):\penalty0 674--690, 1997.

\bibitem[Watkins and Dayan(1992)]{watkins1992q}
C.J.C.H. Watkins and P.~Dayan.
\newblock Q-learning.
\newblock \emph{Machine learning}, 8\penalty0 (3):\penalty0 279--292, 1992.

\end{thebibliography}
\bibliographystyle{plainnat}

\end{document}